\documentclass[10pt]{scrartcl}

\pdfoutput=1

\usepackage[a4paper, total={16cm, 24cm}]{geometry}

\usepackage{amsmath,amssymb}
\usepackage{microtype}
\usepackage[pagebackref]{hyperref}
\hypersetup{
	breaklinks=true,
	colorlinks,
	citecolor=green!45!black,
	linkcolor=red!60!black,
	pagebackref
}

\usepackage{enumitem}
\usepackage{booktabs}
\usepackage{subcaption}
\usepackage{tikz}
\usepackage{algorithm}
\usepackage[noend]{algorithmic}
\usepackage[round]{natbib}
\usetikzlibrary{decorations.pathreplacing, decorations.pathmorphing, positioning, calc}
\numberwithin{equation}{section}

\newcommand{\bfx}{\boldsymbol{x}}

\newcommand{\vv}{\boldsymbol{v}}
\newcommand{\bbR}{\mathbb{R}_{\geq 0}}
\newcommand{\bbN}{\mathbb{N}}

\newcommand{\conv}{\mathrm{conv}}

\newcommand{\half}{\mathrm{half}}
\newcommand{\lefty}{\mathrm{left}}
\newcommand{\calB}{\mathcal{B}}
\newcommand{\calC}{\mathcal{C}}

\newcommand{\shouter}{{s}}
\newcommand{\shoutleft}{{s}_{\lefty}}
\newcommand{\chooser}{{c}}
\newcommand{\Lumpy}{\mathsf{Lumpy}}

\newcommand{\Egal}{\textit{Egal}}
\newcommand{\MMS}{\text{MMS}}
\renewcommand*{\ge}{\geqslant}
\renewcommand*{\geq}{\geqslant}
\renewcommand*{\le}{\leqslant}
\renewcommand*{\leq}{\leqslant}
\renewcommand{\epsilon}{\varepsilon}

\usepackage{amsthm}
\newtheorem{theorem}{Theorem}[section]
\newtheorem{lemma}[theorem]{Lemma}
\newtheorem{proposition}[theorem]{Proposition}
\theoremstyle{definition}
\newtheorem{definition}[theorem]{Definition}
\newtheorem{example}[theorem]{Example}

\usepackage[nameinlink]{cleveref}

\usepackage{authblk}

\begin{document}

\title{Almost Envy-Free Allocations \\ with Connected Bundles\thanks{D.~Peters was supported by ERC grant 639945 (ACCORD). A.~Igarashi is supported by KAKENHI (Grant-in-Aid for JSPS Fellows, No.~18J00997) Japan. While working on this paper, W.~S.~Zwicker was supported by the Oliver Smithies Visiting Fellowship at Balliol College, Oxford. M.~Flammini and C.~Vinci are supported by the Italian MIUR PRIN 2017 Project ALGADIMAR “Algorithms, Games, and Digital Markets”.}}

\date{\vspace{-35pt}}

\author[1]{Vittorio Bil\`o}
\affil[1]{University of Salento, Lecce, Italy}
\author[2]{Ioannis Caragiannis}
\affil[2]{Aarhus University, Aarhus, Denmark}
\author[3]{Michele Flammini}
\affil[3]{Gran Sasso Science Institute, L'Aquila, Italy}
\author[4]{Ayumi Igarashi}
\affil[4]{National Institute of Informatics, Tokyo, Japan}
\author[5]{Gianpiero~Monaco}
\affil[5]{University of L'Aquila, L'Aquila, Italy}
\author[6]{Dominik Peters}
\affil[6]{University of Toronto, Toronto, Canada}
\author[3]{Cosimo Vinci}
\author[7]{William S. Zwicker}
\affil[7]{Union College, Schenectady, USA}

\maketitle

\begin{abstract}
We study the existence of allocations of indivisible goods that are envy-free up to one good (EF1), under the additional constraint that each bundle needs to be connected in an underlying item graph. 
If the graph is a path and the utility functions are monotonic over bundles, we show the existence of EF1 allocations for at most four agents, and the existence of EF2 allocations for any number of agents; our proofs involve discrete analogues of the Stromquist's moving-knife protocol and the Su--Simmons argument based on Sperner's lemma. For identical utilities, we provide a polynomial-time algorithm that computes an EF1 allocation for any number of agents.
For the case of two agents, we characterize the class of graphs that guarantee the existence of EF1 allocations as those whose biconnected 
components are arranged in a path; this property can be checked in linear time.
\end{abstract}

\section{Introduction}\label{sec:intro}

A famous literature considers the problem of \emph{cake-cutting} \citep{BT96a,Robertson1998,Proc15a}. There, a divisible heterogeneous resource (a \emph{cake}, usually formalized as the interval $[0,1]$) needs to be divided among $n$ agents. Each agent has a valuation function over subsets of the cake, usually formalized as an atomless measure over $[0,1]$. The aim is to partition the cake into $n$ pieces, and allocate each piece to one agent, in a ``fair'' way. By fair, we will mean that the allocation is \emph{envy-free}: no agent thinks that another agent's piece is more valuable than her own.

When there are two agents, the classic procedure of cut-and-choose can produce an envy-free division: a knife is moved from left to right, until an agent shouts to indicate that she thinks the pieces to either side are equally valuable. The other agent then picks one of the pieces, leaving the remainder for the shouter. As is easy to see, the result is an envy-free allocation. For three or more agents, finding an envy-free division has turned out to be much trickier. An early result by \citet{Dubins1961} used Lyapunov's Theorem and measure-theoretic techniques to show, non-constructively, that an envy-free allocation always exists. However, as \citet{Stromquist1980} memorably writes, ``their result depends on a liberal definition of a `piece' of cake, in which the possible pieces form an entire $\sigma$-algebra of subsets. A player
who only hopes for a modest interval of cake may be presented instead with a countable union of crumbs.'' In many applications of resource allocation (such as land division, or the allocation of time slots), agents have little use for a severely disconnected piece of cake.

\citet{Stromquist1980} himself offered a solution, and gave a new non-constructive argument (using topology) which proved that there always exists an envy-free division of the cake into \emph{intervals}. Forest Simmons later observed that the proof could be simplified by using Sperner's lemma, and this technique was subsequently presented in a paper by \citet{Su1999}. For the three-agent case, \citet{Stromquist1980} also presented an appealing moving-knife procedure that more directly yields a connected envy-free allocation. For $n\ge 4$ agents, no explicit procedures are known to produce a connected envy-free allocation (i.e., an allocation where the cake is cut in exactly $n-1$ places). However, for $n=4$, several moving-knife procedures exist that only need a few cuts; for example, the Brams--Taylor--Zwicker \citeyearpar{BramsTaylorZwicker1992} procedure requires 11 cuts, and a protocol of \citet{Barbanel2004} requires 5 cuts.

In many applications, the resources to be allocated are not infinitely divisible, and we face the problem of allocating \emph{indivisible goods}. Most of the literature on indivisible goods has not assumed any kind of structure on the item space, in contrast to the rich structure of the interval $[0,1]$ in cake-cutting. Thus, there has been little attention on minimizing the number of ``cuts'' required in an allocation. However, when the items have a spatial or temporal structure, this consideration is important.

In this paper, we study the allocation of items that are arranged on a \emph{path} or other structure, and impose the requirement that only \emph{connected} subsets of items may be allocated to the agents. Formally, we work in the model of \citet{Bouveret2017}, who assume that the items form the vertex set of a graph $G$, and a bundle is connected if it induces a connected subgraph of $G$. For example, such connectivity requirements are encountered in allocation problems in which the items correspond to indivisible pieces of an underlying region of Euclidean space (such as plots of land), and each allocated bundle (i.e., a collection of pieces) must be a connected portion of the space. We are most interested in the case when $G$ is a path, on which connectivity requirements are natural when the items are time slots, for example. In practice, connectivity may not be a hard constraint, and agents may find bundles with few connected components acceptable, but we will not consider such relaxations.

In the work of \citet{Bouveret2017}, it became apparent that techniques from cake-cutting can be usefully ported to achieve good connected allocations in the indivisible case. For example, moving-knife procedures that achieve proportionality in cake-cutting have analogues that produce allocations that satisfy the \emph{maximin share guarantee} \citep{Budi11a}.%
\footnote{Another paper by \citet{Suksompong2017} works in the same model, and also found that procedures for proportionality and other concepts can be applied to the indivisible setting.}

Do envy-free procedures for cake-cutting also translate to the indivisible case? Of course, in general, it is impossible to achieve envy-freeness with indivisibilities (consider two agents and a single desirable item), but we can look for approximations. A relaxation of envy-freeness that has been very influential recently is envy-freeness \emph{up to one good} (EF1), introduced by \citet{Budi11a}. It requires that an agent's envy towards another bundle vanishes if we remove some item from the envied bundle. In the setting without connectivity constraints and with additive valuations, the maximum Nash welfare solution satisfies EF1, as does a simple round-robin procedure \citep{CKM+16a}. The well-known envy-graph algorithm \citep{LMMS04a} also guarantees EF1. However, none of these procedures respects connectivity constraints.

When items are arranged on a path, we prove that connected EF1 allocations exist when there are two, three, or four agents. As was necessary in cake-cutting, we use successively more complicated tools to establish these existence results.
For two agents, there is a discrete analogue of cut-and-choose that satisfies EF1. In that procedure, a knife moves across the path, and an agent shouts when the knife reaches what we call a \emph{lumpy tie}, that is when the bundles to either side of the knife have equal value \emph{up to one item}. For three agents, we design an algorithm mirroring Stromquist's moving-knife procedure which guarantees EF1. For four agents, we show that Sperner's lemma can be used to prove that an EF1 allocation exists, via a technique inspired by the Simmons--Su approach, and an appropriately triangulated simplex of connected partitions of the path. For five or more agents, we were not able to establish the existence of EF1 allocations on a path, but we can show (again via Sperner's lemma) that EF2 allocations exist, strengthening a prior result of \citet{Suksompong2017}. We also show that if all agents have the same valuation function over bundles, then an egalitarian-welfare-optimal allocation, after suitably reallocating some items, is EF1.

These existence results require only that agents' valuations are monotonic (they need not be additive), and in addition,  ensure that the constructed allocation satisfies the maximin share guarantee (see Appendix~\ref{sec:mms}). Moreover, the fairness guarantee of our algorithms is slightly stronger than the standard notion of EF1: in the returned allocations, envy can be avoided by removing just an \emph{outer} item -- one whose removal leaves the envied bundle connected. Computationally speaking, all our existence results are constructive in the weak sense that an EF1 allocation can be found by iterating through all $O(m^n)$ connected allocation (this stands in contrast to cake-cutting where we cannot iterate through all possibilities). While we know of no faster algorithms to obtain an EF1 or EF2 allocation in the cases where we appeal to Sperner's lemma, our other procedures (for two or three agents, or for identical valuations) can all be implemented efficiently to produce a fair allocation in polynomial time. We summarize our results concerning paths in Table \ref{table}.

\begin{table*}[t]
	\centering
	\begin{tabular}{llll}
		\toprule
		& $\#$ of agents & EF1 & EF2 \\
		\midrule
		Existence
		& $n=2$   & \checkmark~(Thm.~\ref{lemma:cutchoose-path}) & \checkmark\\
		& $n=3$ & \checkmark~(Thm.~\ref{thm:EF1-3agents})& \checkmark\\
		& $n=4$ & \checkmark~(Thm.~\ref{thm:EF1-4agents}) & \checkmark\\
		& $n \geq 5$ & open problem & \checkmark~(Thm.~\ref{thm:EF2})\\
		& $n$ agents, identical valuations & \checkmark~(Thm.~\ref{thm:ef1-identical}) & \checkmark \\
		\midrule
		Complexity
		& $n=2$ & $O(\log m)$~\citep{Oh2018} & $O(\log m)$\\
		& $n=3$ & $O(m)$~(Thm.~\ref{thm:EF1-3agents}) &  $O(m)$ \\
		& $n \geq 4$ &  $O(m^n)$& $O(m^n)$\\
		& $n$ agents, identical valuations & $O(mn)$~(Thm.~\ref{thm:ef1-identical}) & $O(mn)$ \\
		\bottomrule
	\end{tabular}
	\caption{
		Overview of our results for paths. Here, $n$ denotes the number of agents and $m$ the number of items. Agents' valuations are assumed to be monotone. The mark $\checkmark$ represents that a connected allocation satisfying the corresponding fairness notion exists. When no reference is given, the result follows from other results in the table. 
	}
	\vspace{-2pt}
	\label{table}
\end{table*}

In simultaneous and independent work, \citet{Oh2018} designed protocols to find EF1 allocations in the setting without connectivity constraints, aiming for low \emph{query complexity}. They found that adapting cake-cutting protocols to the setting of indivisible items arranged on a path is an especially potent way to achieve low query complexity. This led them to also study a discrete version of the cut-and-choose protocol which achieves connected EF1 allocations for two agents, and they found an alternative proof that an EF1 allocation on a path always exists with identical valuations. They also present a discrete analogue of the Selfridge--Conway procedure which, for three agents with additive valuations, produces an allocation of a path into bundles that have a constant number of connected components. However, they do not study connected allocations on graphs that are not paths, and they do not consider the case of (non-identical) general valuations with more than two agents.

A recurring theme in our algorithms is the specific way that the moving knives from cake-cutting are rendered in the discrete setting. While one might expect knives to be placed over the edges of the path, and `move' from edge to edge, we find that this movement is too `fast' to ensure EF1 (see also footnote~\ref{footnote:weak-ef1} regarding EF2). Instead, our knives alternate between hovering over edges and items. When a knife hovers over an item, we imagine the knife's blade to be `thick': the knife \emph{covers} the item, and agents then pretend that the covered item does not exist. These intermediate steps are useful, since they can tell us that envy will vanish if we hide an item from a bundle.

What about graphs $G$ other than paths? Our existential 
results for paths immediately generalize to traceable graphs (those that contain a Hamiltonian path), since we can run the algorithms pretending that the graph only consists of the Hamiltonian path. For the two-agent case, we completely characterize the class of graphs that guarantee the existence of EF1 allocations: Our discrete cut-and-choose protocol can be shown to work on all graphs $G$ that admit a \emph{bipolar numbering}, which exists if and only if the biconnected components (blocks) of $G$ can be arranged in a path. By constructing counterexamples, we prove that no graph failing this condition (for example, a star) guarantees EF1, even for identical, additive, binary valuations. For the case of three or more agents, it is a challenging open problem to characterize the class of graphs guaranteeing EF1 (or even to find an infinite class of non-traceable graphs that guarantees EF1). 

\section{Preliminaries}\label{sec:prem}
For each natural number $s \in \bbN$, write $[s]=\{1,2,\ldots,s\}$. Let $N=[n]$ be a finite set of \emph{agents} and $G=(V,E)$ be an undirected finite graph, where $V=\{v_1,v_2,\ldots,v_m\}$. 
The vertices in $V$ correspond to  \emph{items}.
A subset $I$ of $V$ is \emph{connected} if it induces a connected subgraph of $G$.
We write $\calC(V)$ for the set of connected subsets of $V$.
We call a set $I \in \calC(V)$ a (connected) \emph{bundle}.
Each agent $i \in N$ has a \emph{valuation function} $u_i: \calC(V) \rightarrow \mathbb R$ over connected bundles, which we will always assume to be \emph{monotonic}, that is, $X \subseteq Y$ implies $u_i(X) \leq u_i(Y)$.
We also assume that $u_i(\emptyset)=0$ for each $i \in N$.
Monotonicity implies that items are \emph{goods}; we do not consider bads (or chores) in this paper.
We say that an agent $i \in N$ \emph{weakly prefers} bundle $X$ to bundle $Y$ if $u_i(X) \geq u_i(Y)$.%
\footnote{Our arguments only operate based on agents' ordinal preferences over bundles, and the (cardinal) valuation functions are only used for notational convenience. One exception, perhaps, is in Algorithm~\ref{alg:leximin-ef1} where we calculate a leximin allocation, but the algorithm can be applied after choosing an arbitrary utility function consistent with the ordinal preferences.}
A (connected) \emph{allocation}~$A : N \to \calC(V)$ assigns each agent $i\in N$ a connected bundle $A(i) \in \calC(V)$ such that each item occurs in exactly one bundle, i.e., $\bigcup_{i\in N} A(i) = V$ and $A(i) \cap A(j) = \emptyset$ when $i\neq j$. We often write $I^i$ for the bundle $A(i)$ assigned to agent $i$.

We say that the agents have \emph{identical valuations} if, for all $i,j\in N$ and every bundle $I\in \calC(V)$, we have $u_i(I)=u_j(I)$. A valuation function $u_i$ is \emph{additive} if $u_i(I) = \sum_{v\in I} u_i(\{v\})$ for each bundle $I\in \calC(V)$. Many examples in this paper will use identical additive valuations, and will take $G$ to be a path. In this case, we use a shorthand to specify these examples; the meaning of this notation should be clear. For example, we write ``2--1--3--1'' to denote an instance with four items $v_1,v_2,v_3,v_4$ arranged on a path, and where $u_i(\{v_1\}) =2$, \dots, $u_i(\{v_4\}) = 1$ for each $i$. For such an instance, an allocation will be written as a tuple, e.g., (2, 1--3--1) denoting an allocation allocating bundles $\{v_1\}$ and $\{v_2, v_3, v_4\}$, noting that with identical valuations it does not usually matter which agent receives which bundle.

An allocation $A$ is \emph{envy-free} if $u_i(A(i)) \ge u_i(A(j))$ for every pair $i,j\in N$ of agents, that is, if every agent thinks that their bundle is at least as good as any other bundle in the allocation.
It is well-known that an envy-free allocation may not exist (consider two agents and one good).
The main fairness notion that we study is a version of \emph{envy-freeness up to one good} (EF1), a relaxation of envy-freeness introduced by \citet{Budi11a}, adapted to the model with connectivity constraints.
This property states that an agent $i$ will not envy another agent $j$ after we remove some item from $j$'s bundle. Since we only allow connected bundles in our set-up, we may only remove an item from $A(j)$ if removal of this item leaves the bundle connected.

\begin{definition}[EF1: envy-freeness up to one \emph{outer} good]
	An allocation $A$ satisfies \emph{EF1} if, for any pair $i,j \in N$ of agents, either $A(j) = \emptyset$ or there is a good $v \in A(j)$ such that $A(j) \setminus \{v\}$ is connected and $u_i(A(i)) \ge u_i(A(j)\setminus \{v\})$.
\end{definition}

In the instance 2--1--3--1 for two agents, the allocation (2--1, 3--1) is EF1, since the left agent's envy can be eliminated by removing the item of value 3 from the right-hand bundle. However, the allocation (2, 1--3--1) fails to be EF1 according to our definition, since eliminating either outer good of the right bundle does not prevent envy.%
\footnote{This example shows that our definition is strictly stronger than the standard definition of EF1 without connectivity constraints. In the instance 2--1--3--1, considered without connectivity constraints, the allocation (2, 1--3--1) does satisfy EF1 since in the standard setting we are allowed to remove the middle item (with value 3) of the right bundle.}

\begin{definition}
	A graph $G$ \emph{guarantees EF1} for $n$ agents if, for all possible monotonic valuations for $n$ agents, there exists some connected allocation that is EF1. A graph $G$ guarantees EF1 for $n$ agents and a restricted class of valuations if, for all allowed valuations, a connected EF1 allocation exists.
\end{definition}

For reasoning about EF1 allocations, let us introduce a few shorthands. Given an allocation $A$ we will say that $i\in N$ \emph{does not envy $j\in N$ up to $v$} if $u_i(A(i)) \ge u_i(A(j)\setminus \{v\})$.
The \emph{up-to-one valuation} $u^-_i: \calC(V) \rightarrow \bbR$ of agent $i\in N$ is defined, for every $I \in \calC(V)$, as
\begin{equation}
\label{eq:up-to-one-valuation}
u^-_i(I) :=
\begin{cases}
0 & \text{if $I = \emptyset$,} \\
\min \big\{ u_i(I\setminus\{v\}) : v\in I \text{ such that } I\setminus\{v\}\text{ is connected} \big\} & \text{if $I \neq \emptyset$.}
\end{cases}
\end{equation}
Thus, an allocation $A$ satisfies EF1 if and only if $u_i(A(i)) \ge u_i^-(A(j))$ for any pair $i,j\in N$ of agents.

As we show in the appendix in Example~\ref{ex:EFX}, allocations satisfying a strengthened version of EF1 called envy-freeness up to the \emph{least} good (EFX) \citep{CKM+16a} may not exist on a path.

Given an ordered sequence of the vertices $P=(v_1,v_2,\ldots,v_m)$, and $j,k\in [m]$ with $j\leq k$, we write $P(v_j,v_k)$ for the subsequence from $v_j$ to $v_k$, so
$P(v_j,v_k)=(v_j,v_{j+1},\ldots, v_{k-1},v_{k})$. With a little abuse of notation, we often identify a subsequence $P(v_j,v_k)$ with the bundle of the corresponding vertices. 
Let $L(v_j)=P(v_1,v_{j-1})$ be the subsequence of vertices strictly left of $v_j$ and $R(v_j)=P(v_{j+1},v_{m})$ be the subsequence of vertices strictly right of $v_j$. 
When graph $G$ is a path, we always implicitly assume that its vertices $v_1,v_2,\ldots,v_m$ are numbered from left to right according to the order they appear along the path, so that the set of the edges of $G$ is $\{\{v_j,v_{j+1}\} : 1 \leq j < m\}$. Each connected bundle in the path clearly corresponds to a subpath or subsequence of the vertices.
A \emph{Hamiltonian path} of a graph $G$ is a path that visits all the vertices of the graph exactly once. A graph is \emph{traceable} if it contains a Hamiltonian path.

\section{EF1 existence for two agents}\label{sec:two}
In cake-cutting for two agents, the standard way of obtaining an envy-free allocation is the cut-and-choose protocol: Alice divides the cake into two equally-valued pieces, and Bob selects the piece he prefers; the other piece goes to Alice. The same strategy almost works in the indivisible case when items form a path; the problem is that Alice might not be able to divide the items into two exactly equal pieces. Instead, we ask Alice to divide the items into pieces that are equally valued ``up to one good''. The formal version is as follows. For a sequence of vertices $P=(v_1,v_2,\ldots,v_m)$ and an agent $i$, we say that $v_j$ is the \emph{lumpy tie} over $P$ for agent $i$ if $j$ is the smallest index such that
\begin{equation}
\label{eq:def-lumpy-tie-1}
u_i(L(v_j) \cup \{ v_j \}) \ge u_i(R(v_j)) \quad \text{and}\quad u_i(R(v_j) \cup \{v_j\}) \ge u_i(L(v_j)).
\end{equation}
For example, when $i$ has additive valuations 1--3--2--1--3--1, then the third item (of value 2) is the lumpy tie for $i$, since $1+3+2 \ge 1 + 3 + 1$ and $2+1+3+1 \ge 1 + 3$. The lumpy tie always exists: taking $j$ to be the smallest index such that $u_i(L(v_j) \cup \{ v_j \}) \ge u_i(R(v_j))$ (which exists as the inequality holds for $j=m$ by monotonicity), the first part of \eqref{eq:def-lumpy-tie-1} holds. If $j = 1$, the second part of \eqref{eq:def-lumpy-tie-1} is immediate by monotonicity. If $j > 1$, then since $j$ is minimal, we have $u_i(L(v_j)) = u_i(L(v_{j-1}) \cup \{ v_{j-1} \}) < u_i(R(v_{j-1})) = u_i(R(v_j) \cup \{ v_j \})$ as required.

Using lumpy ties, our discrete version of the cut-and-choose protocol is specified as follows.

\begin{definition}
	The \textbf{discrete cut-and-choose protocol for $n=2$ agents} on a sequence $P=(v_1,v_2,\ldots,v_m)$ proceeds as follows:
	\begin{itemize}
		\item \textit{Step 1.} Alice selects her lumpy tie $v_j$ over $(v_1,v_2,\ldots,v_m)$.
		\item \textit{Step 2.} Bob chooses a weakly preferred bundle among $L(v_j)$ and $R(v_j)$.
		\item \textit{Step 3.} Alice receives the bundle of all the remaining vertices, including $v_j$.
	\end{itemize}
\end{definition}

Intuitively, the protocol allows Alice to select an item $v_j$ that she will receive for sure, with the advice that the two pieces to either side of $v_j$ should have almost equal value to her. Then, Bob is allowed to choose which side of $v_j$ he wishes to receive. In our example with valuations 1--3--2--1--3--1, Alice selects the lumpy tie of value 2, then Bob chooses the bundle 1--3--1 to the right and receives it, and Alice receives the bundle 1--3--2. The result is EF1.
This is true in general, and also if valuations are not identical.

\begin{proposition}\label{lemma:cutchoose-path}
	When $G$ is a path and there are $n=2$ agents, the discrete cut-and-choose protocol yields an EF1 allocation.
\end{proposition}
\begin{proof}
	Clearly, the protocol returns a connected allocation.
	The returned allocation satisfies EF1: Bob does not envy Alice up to item $v_j$, since Bob receives his preferred bundle among $L(v_j)$ and $R(v_j)$. Also, by \eqref{eq:def-lumpy-tie-1}, Alice does not envy Bob, since Alice either receives the bundle $L(v_j) \cup \{ v_j \}$ which she weakly prefers to Bob's bundle $R(v_j)$, or she receives the bundle $R(v_j) \cup \{ v_j \}$, which she weakly prefers to Bob's bundle $L(v_{j})$.
\end{proof}

Proposition~\ref{lemma:cutchoose-path} implies that an EF1 allocation always exists on a path. Hence, an EF1 allocation exists for every traceable graph $G$: simply use the discrete cut-and-choose protocol on a Hamiltonian path of $G$. In fact, the discrete cut-and-choose protocol works on a broader class of graphs: We only need to require that the vertices of the graph can be numbered in a way that the allocation resulting from the discrete cut-and-choose protocol is guaranteed to be connected. Since the protocol always partitions the items into an initial and a terminal segment of the sequence, such a numbering needs to satisfy the following property.

\begin{definition}\label{def:quasi-Hamiltonian}
	A \emph{bipolar numbering} of a graph $G$ is an ordering $(v_1, v_2, \dots ,$ $v_m)$ of its vertices such that for all $j \in [n]$, the sets $L(v_j) \cup \{ v_j \}$ and $R(v_j) \cup \{ v_j \}$ are connected in $G$.
\end{definition}

In a slightly different context, bipolar numberings are known as $st$-numberings and turn out to be useful in algorithms for testing planarity and for graph drawing \citep{Lempel1967,Even1976,Tarjan1986}.
The more common (equivalent) definition is phrased to say that a numbering is bipolar if, for every $j\in [n]$, the vertex $v_j$ has a neighbor that appears earlier in the sequence, and a neighbor that appears later in the sequence.  

Clearly, every traceable graph has a bipolar numbering, since we can just use a Hamiltonian path. However, there are also non-traceable graphs that admit a bipolar numbering. Figure~\ref{fig:qHgraph} shows some examples.

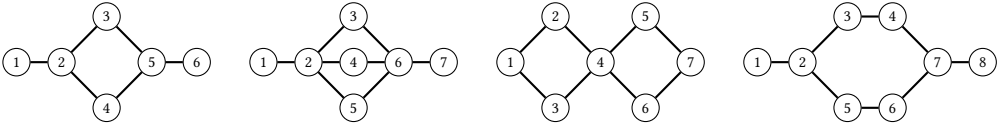
\begin{figure*}[ht]
	\centering
		\begin{minipage}{0.3\textwidth}
			\begin{tikzpicture}[scale=0.9, transform shape]
			\node[draw, circle](1) at (0,0) {1};
			\node[draw, circle](2) at (1,0) {2};
			\node[draw, circle](3) at (2,1) {3};
			\node[draw, circle](4) at (2,-1) {4};
			\node[draw, circle](5) at (3,0) {5};
			\node[draw, circle](6) at (4,0) {6};
			\draw[-, >=latex,thick] (1)--(2) (2)--(3) (2)--(4) (3)--(5) (4)--(5) (5)--(6);	
			\end{tikzpicture}
		\end{minipage}
		\begin{minipage}{0.3\textwidth}
			\begin{tikzpicture}[scale=0.9, transform shape]
			\node[draw, circle](1) at (0,0) {1};
			\node[draw, circle](2) at (1,0) {2};
			\node[draw, circle](3) at (2,1) {3};
			\node[draw, circle](4) at (2,0) {4};
			\node[draw, circle](5) at (2,-1) {5};
			\node[draw, circle](6) at (3,0) {6};
			\node[draw, circle](7) at (4,0) {7};
			\draw[-, >=latex,thick] (1)--(2) (2)--(3) (2)--(4) (2)--(5) (3)--(6) (4)--(6) (5)--(6) (6)--(7);	
			\end{tikzpicture}
		\end{minipage}
		\begin{minipage}{0.3\textwidth}
			\begin{tikzpicture}[scale=0.8, transform shape]
			\node[draw, circle](1) at (0,0) {1};
			\node[draw, circle](2) at (1,0) {2};
			\node[draw, circle](3) at (2,1) {3};
			\node[draw, circle](4) at (3,1) {4};
			\node[draw, circle](5) at (2,-1) {5};
			\node[draw, circle](6) at (3,-1) {6};
			\node[draw, circle](7) at (4,0) {7};
			\node[draw, circle](8) at (5,0) {8};
			
			\draw[-, >=latex,thick] (1)--(2) (2)--(3) (3)--(4) (2)--(5) (5)--(6) (4)--(7) (6)--(7) (7)--(8);	
			\end{tikzpicture}
		\end{minipage}
	\caption{Non-traceable graphs with bipolar numberings.
		\label{fig:qHgraph}
	}
\end{figure*}

\begin{proposition}\label{lemma:cutchoose}
	When there are $n=2$ agents, then the discrete cut-and-choose protocol run on a bipolar numbering of $G$ yields an EF1 allocation.
\end{proposition}
\begin{proof}
	The discrete cut-and-choose protocol returns an allocation whose bundles are either initial or terminal segments of the ordered sequence $(v_1, v_2, \dots , v_m)$. By definition of a bipolar numbering, such an allocation is connected, and it is EF1 by the same argument as in Proposition~\ref{lemma:cutchoose-path}.
\end{proof}

It is clear that the discrete cut-and-choose protocol cannot be extended to graphs other than those admitting a bipolar numbering. However, it could be that a different protocol is able to produce EF1 allocations on other graphs. In the remainder of this section, we prove that this is not the case: for $n=2$ agents, a connected graph $G$ guarantees the existence of an EF1 allocation if and only if it admits a bipolar numbering. This completely characterizes the class of graphs that guarantee EF1 existence in the two-agent case.\footnote{Note that no non-trivial disconnected graph guarantees EF1 for two agents: If $G$ is disconnected, take a connected component $C$ with at least two vertices. Let both agents have additive valuations that value each item in $C$ at 1, and value items outside of $C$ at 0. Then, in a connected allocation, all items in $C$ must go to a single agent, since the other agent needs to receive items from another connected component. This induces envy in the other agent that is not bounded by one good.}

For a different number of agents, the class of graphs guaranteeing an EF1 allocation will be different. In particular, the star with three leaves does not guarantee an EF1 allocation for two agents (as it does not have a bipolar numbering, see below), but one can check that this star does guarantee an EF1 allocation for three or more agents (see Example \ref{ex:EF1:nonHamiltonian} in the appendix).

\subsection{Characterization of graphs guaranteeing EF1 for two agents}
Based on a known characterization of graphs admitting a bipolar numbering,
we characterize this class in terms of forbidden substructures. We then show that these forbidden structures are also forbidden for EF1: if a graph contains such a structure, we can exhibit an additive valuation profile for which no EF1 allocation exists.

As a simple example, consider the star with three leaves, which is the smallest connected graph that does not have a bipolar numbering. 
\[  
\scalebox{0.68}{
	\begin{tikzpicture}[scale=1, transform shape]
		\node[draw, circle](1) at (0,0) {};
		\node[draw, circle](2) at (-1,-0.7) {};
		\node[draw, circle](3) at (0,-1.2) {};
		\node[draw, circle](4) at (1,-0.7) {};
		\draw[-, >=latex,thick] (1)--(2) (1)--(3) (1)--(4);	
\end{tikzpicture}}
\]
Take two agents with identical additive valuations that value each item at 1. Any connected allocation must allocate three items to one agent, and a single item to the other agent. Then the latter agent envies the former agent, even up to one good. This star is an example of a forbidden substructure called a trident, which takes one of two forms, illustrated in Figure~\ref{fig:trident}.
\begin{definition}
A graph $G$ \textit{contains a trident} if either
\begin{enumerate}
	\item[(a)] there is a vertex $s$ whose removal from $G$ leaves three or more connected components (a \emph{type I} trident), or
	\item[(b)] there are subgraphs $C, P_1, P_2, P_3$ of $G$ such that (i)
	 $P_1,P_2,P_3$ are vertex-disjoint, (ii) each $P_i$ contains at least two vertices, (iii) $C$ has exactly one \emph{contact vertex} $s_i$ in common with $P_i$, $i=1,2,3$, and (iv) for $i=1,2,3$, removal of vertex $s_i$ from $G$ disconnects $P_i \setminus \{ s_i \}$ from $C \setminus \{ s_i \}$ (hence from the other two $P_j$) in $G$ (a \emph{type II} trident).
\end{enumerate}
\end{definition}

	\begin{figure*}[ht]
		\centering
		\begin{subfigure}[b]{0.45\textwidth}
			\centering
			\begin{tikzpicture}[scale=0.8, transform shape]
			\def \radius {2cm}
			\node[draw, circle,fill=black!80, inner sep=0.75mm](node0) at (0,0) {};
			
			\node(node1) at ({60}:\radius) {};
			\node(node2) at ({120}:\radius) {};
			\node(node3) at ({180}:\radius) {};
			\node(node4) at ({240}:\radius) {};
			\node(node5) at ({300}:\radius) {};
			\node(node6) at ({360}:\radius) {};
			
			\draw[decorate, decoration={bent}] (node0) --({60}:\radius) (node0) --({120}:\radius) (node0) --({180}:\radius) (node0) --({240}:\radius) (node0) --({300}:\radius) (node0) --({360}:\radius);
			\draw[decorate, decoration={random steps,segment length=5mm}] ({60}:\radius) -- ({120}:\radius) ({180}:\radius) -- ({240}:\radius) ({300}:\radius) -- ({360}:\radius);
			
			
			\end{tikzpicture}
		\end{subfigure}
		\quad
		\begin{subfigure}[b]{0.45\textwidth}
			\centering
			\begin{tikzpicture}[scale=0.4, transform shape,
			zig/.style={decorate, decoration={zigzag, amplitude=0.05mm}}]
			
			\draw[zig] (1,0.2) arc (180:90:2cm);
			\draw[zig] (1,-0.2) arc (180:360:2cm);
			\draw[zig] (3.1,2.2) arc (90:0:2cm);
			
			\node[draw, circle,fill=black!80](1) at (1,0) {};
			\node[draw, circle,fill=black!80](2) at (3.1,2.2) {};
			\node[draw, circle,fill=black!80](3) at (5.1,0) {};
			\node at (3,0) {\huge $C$};
			
			\begin{scope}[shift={(-4.1,0)}]
			\draw[zig] (1,0) arc (180:0:2cm);
			\draw[zig] (1,0) arc (180:360:2cm);
			\node at (3,0) {\huge $P_1$};
			\end{scope}
			
			\begin{scope}[shift={(4.2,0)}]
			\draw[zig] (1,0) arc (180:0:2cm);
			\draw[zig] (1,0) arc (180:360:2cm);
			\node at (3,0) {\huge $P_3$};
			\end{scope}
			
			\begin{scope}[shift={(0,4.2)}]
			\draw[zig] (1,0) arc (180:0:2cm);
			\draw[zig] (1,0) arc (180:360:2cm);
			\node at (3,0) {\huge $P_2$};
			\end{scope}
			
			\end{tikzpicture}
		\end{subfigure}
		\caption{A type I trident (left) and a type II trident (right).}
		\label{fig:trident}
	\end{figure*}

\noindent
We will prove that a graph $G$ fails to admit a bipolar numbering, and fails to guarantee EF1 for two agents, if and only  if $G$ contains a trident.   To reason about these structures, it is useful to consider the standard concept of the \emph{block decomposition} of a graph \citep[see, e.g., the textbook][Sec.~5.2]{Bondy:2008}. 

\begin{definition}
	A \emph{decomposition} of a graph $G=(V,E)$ is a family $\{F_1,F_2,\ldots,F_t\}$ of edge-disjoint subgraphs of $G$ such that $\bigcup^t_{i=1}E(F_i)=E$ where $E(F_i)$ is the set of edges of $F_i$.
	A vertex is called a \emph{cut vertex} of a graph $G$ if removing it increases the number of connected components of $G$. A graph $G$ is \emph{biconnected} if $G$ is connected and does not have a cut vertex. A \emph{block} of $G$ is a maximal biconnected subgraph of $G$.
\end{definition}

Equivalently, a block of a graph $G$ can be defined as a maximal subgraph of $G$ where each pair of vertices lie on a common cycle \citep{Bondy:2008}. Given a connected graph $G$, we define a bipartite graph $B(G)$ with bipartition $(\calB,S)$, where $\calB$ is the set of blocks of $G$ and $S$ is the set of cut vertices of $G$; a block $B$ and a cut vertex $v$ are adjacent in $B(G)$ if and only if $B$ includes $v$. Since every cycle of a graph is included in some block, the graph $B(G)$ is a tree:

\begin{lemma}[e.g., \citealp{Bondy:2008}, Prop.~5.3]
	Let $G$ be a connected graph. Then
	\begin{itemize}
		\item any two blocks of $G$ have at most one cut vertex in common;
		\item the set of blocks forms a decomposition of $G$; and
		\item the graph $B(G)$ is a tree.
	\end{itemize}
\end{lemma}

Thus, for a connected graph $G$, we call $B(G)$ the \emph{block tree} of $G$. It turns out that $G$ admits a bipolar numbering if and only if $B(G)$ is a path. For example, the graphs shown in Figure~\ref{fig:qHgraph} all have their blocks arranged in a path (so that $B(G)$ is a path), as shown in Figure~\ref{fig:qHgraph-blocks}.

\begin{figure*}[ht]
	\centering
		\begin{minipage}{0.3\textwidth}
			\begin{tikzpicture}[scale=0.8, transform shape]
			\draw[line width=0.4mm,gray,dotted] (0.5,0) ellipse (28pt and 20pt);
			\draw[line width=0.4mm,gray,dotted] (2,0) ellipse (40pt and 40pt);
			\draw[line width=0.4mm,gray,dotted] (3.5,0) ellipse (28pt and 20pt);
			
			\node[draw, circle](1) at (0,0) {1};
			\node[draw, circle](2) at (1,0) {2};
			\node[draw, circle](3) at (2,1) {3};
			\node[draw, circle](4) at (2,-1) {4};
			\node[draw, circle](5) at (3,0) {5};
			\node[draw, circle](6) at (4,0) {6};
			\draw[-, >=latex,thick] (1)--(2) (2)--(3) (2)--(4) (3)--(5) (4)--(5) (5)--(6);	
			\end{tikzpicture}
		\end{minipage}
		\begin{minipage}{0.3\textwidth}
			\begin{tikzpicture}[scale=0.8, transform shape]
			\draw[line width=0.4mm,gray,dotted] (0.5,0) ellipse (28pt and 20pt);
			\draw[line width=0.4mm,gray,dotted] (2,0) ellipse (40pt and 40pt);
			\draw[line width=0.4mm,gray,dotted] (3.5,0) ellipse (28pt and 20pt);
			
			\node[draw, circle](1) at (0,0) {1};
			\node[draw, circle](2) at (1,0) {2};
			\node[draw, circle](3) at (2,1) {3};
			\node[draw, circle](4) at (2,0) {4};
			\node[draw, circle](5) at (2,-1) {5};
			\node[draw, circle](6) at (3,0) {6};
			\node[draw, circle](7) at (4,0) {7};
			\draw[-, >=latex,thick] (1)--(2) (2)--(3) (2)--(4) (2)--(5) (3)--(6) (4)--(6) (5)--(6) (6)--(7);	
			\end{tikzpicture}
		\end{minipage}
		\begin{minipage}{0.3\textwidth}
			\begin{tikzpicture}[scale=0.8, transform shape]
			\draw[line width=0.4mm,gray,dotted] (0.5,0) ellipse (28pt and 20pt);
			\draw[line width=0.4mm,gray,dotted] (2.5,0) ellipse (54pt and 42pt);
			\draw[line width=0.4mm,gray,dotted] (4.5,0) ellipse (28pt and 20pt);
			
			\node[draw, circle](1) at (0,0) {1};
			\node[draw, circle](2) at (1,0) {2};
			\node[draw, circle](3) at (2,1) {3};
			\node[draw, circle](4) at (3,1) {4};
			\node[draw, circle](5) at (2,-1) {5};
			\node[draw, circle](6) at (3,-1) {6};
			\node[draw, circle](7) at (4,0) {7};
			\node[draw, circle](8) at (5,0) {8};
			
			\draw[-, >=latex,thick] (1)--(2) (2)--(3) (3)--(4) (2)--(5) (5)--(6) (4)--(7) (6)--(7) (7)--(8);	
			\end{tikzpicture}
		\end{minipage}
	\caption{Block decompositions of the graphs in Figure \ref{fig:qHgraph}.}
	\label{fig:qHgraph-blocks}
\end{figure*}
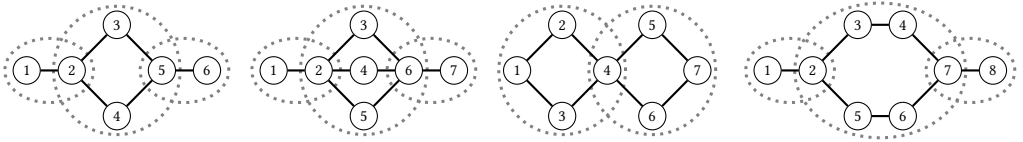

\begin{lemma}\label{lemma:block-trees-quasi-Hamiltonian}
	A graph $G$ admits a bipolar numbering if its block tree $B(G)$ is a path.
\end{lemma}
\begin{proof}
	\citet{Lempel1967} show that $G$ admits a bipolar numbering if there are $s,t \in V$ such that adding an edge $\{s,t\}$ to $G$ makes it biconnected. If $B(G)$ is a path, let $B_1$ and $B_2$ be the leaf blocks at the ends of the path $B(G)$. Take any $s\in B_1$ and $t\in B_2$. If we add the edge $\{s,t\}$ to $G$, the graph becomes biconnected. Hence, $G$ admits a bipolar numbering.
\end{proof}

There is a linear-time algorithm based on depth-first search to construct a bipolar numbering for any biconnected graph \citep{Even1976,Tarjan1986},  and one can also calculate the block tree $B(G)$ of a given graph in linear time \citep{Hopcroft1973}. Thus, in linear time, we can compute a bipolar numbering of a graph or report that none exists. Clearly, given a bipolar numbering, the discrete cut-and-choose protocol can also be run in linear time.

Next, we show that if $B(G)$ is not a path, then $G$ cannot guarantee EF1. The proof constructs explicit counter-examples, which have a very simple structure. We say that additive valuations $u_i$ are \emph{binary} if $u_i(\{ v \}) \in \{0,1\}$ for every $v\in V$.

\begin{lemma}\label{lemma:EF1-block-trees}
	Let $G$ be a connected graph.  
\begin{itemize}
	\item If the block tree $B(G)$ of $G$ is \emph{not} a path, then $G$ contains a trident.
	\item If $G$ contains a trident, then there exist identical, additive, binary valuations over $G$ for two agents such that no connected allocation is EF1.	
\end{itemize}
\end{lemma}
\begin{proof}
If $B(G)$ is not a path, then it contains a vertex with at least three neighbors, and thus either
\begin{itemize}
\item[(a)]\label{case:cut} there is a cut vertex $s$ adjacent to three blocks $B_1$, $B_2$, and $B_3$; or
\item[(b)]\label{case:block} there is a block $B$ adjacent to three different cut vertices $s_1$, $s_2$, and $s_3$.
\end{itemize}
Note that in both cases, all blocks contain at least two vertices each, as maximality guarantees that a block in a connected graph $G$ never consists of a single vertex, unless $G$ itself has only one vertex. 
Thus, in case (a), $G$ contains a type I trident.
In case (b), the cut vertices $s_1$, $s_2$, and $s_3$ serve as the contact vertices in the earlier definition of type II tridents and are adjacent to blocks that serve as the subgraphs $P_1$, $P_2$, and $P_3$. This proves the first part. 

To prove the second part, we construct identical additive valuations that do not admit an EF1 allocation. If $G$ contains a type I trident, let $s$ be the corresponding cut vertex, and choose vertices $v_1 , v_2, v_3$ from each of three different connected components that remain after $s$ is deleted from $G$. 
The two agents have utility $1$ for each of $s$, $v_1$, $v_2$, and $v_3$, and $0$ for the remaining vertices. 
Now take any connected allocation $(I_1, I_2)$.
One of the bundles, say $I_1$, includes the cut vertex $s$. 
Then $I_2$ can contain at most one of the vertices $v_1$, $v_2$, $v_3$, since $I_2$ is connected and does not contain $s$ yet any path between distinct $v_i$ and $v_j$ goes through $s$.
Hence $u_i(I_2) \le 1$. Now, the bundle $I_1$ contains $s$ and at least two of $v_1$, $v_2$, $v_3$, so $u_i(I_1) \ge 3$.
Thus, the allocation is not EF1.

Suppose $G$ contains a type II trident consisting of subgraphs $C, P_1, P_2, P_3$ with contact vertices $s_1, s_2, s_3$.  Then for $i=1,2,3$ choose a vertex $v_i \neq s_i$ from $P_i$. 
The two agents have utility $1$ for each of $s_1$, $s_2$, $s_3$, $v_1$, $v_2$, and $v_3$, and $0$ for the remaining vertices. 
Now take any connected allocation $(I_1, I_2)$.
One of the bundles, say $I_1$, contains at least two contact vertices $s_i$ and the other contains at most one contact vertex $s_i$. 
Say that $s_1, s_2 \in I_1$.
Now, $G\setminus\{ s_1,s_2 \}$ has at least three connected components, and since $I_2$ is connected, it must be contained in one of these components.
But each component contains at most two vertices with utility 1, so $u_i(I_2) \le 2$. 
Since there are six vertices with utility 1 in total, $u_i(I_1) \ge 4$.
Thus, the allocation is not EF1.
\end{proof}

Combining these results, we obtain the promised characterization.

\begin{theorem}  \label{thm:EQ} The following conditions are equivalent for every connected graph $G$:
	\begin{enumerate}
		\item $G$ admits a bipolar numbering.
		\item $G$ guarantees EF1 for two agents.
		\item $G$ guarantees EF1 for two agents with identical, additive, binary valuations.
		\item $G$ does not contain a trident.
		\item The block tree $B(G)$ is a path.
	\end{enumerate}
\end{theorem}

\begin{proof}
	The implication $(1) \Rightarrow (2)$ follows from Proposition \ref{lemma:cutchoose} which shows that the discrete cut-and-choose protocol yields a connected EF1 allocation when run on a bipolar numbering.
	The implication $(2) \Rightarrow (3)$ is immediate.
	The implications $(3) \Rightarrow (4)$ and $(4) \Rightarrow (5)$ follow from Lemma \ref{lemma:EF1-block-trees} which proves the contrapositives.
	Finally, $(5) \Rightarrow (1)$ follows from Lemma \ref{lemma:block-trees-quasi-Hamiltonian}.
\end{proof}

The equivalence $(2) \Leftrightarrow (3)$ is noteworthy and perhaps surprising: It is often easier to guarantee fairness when agents' valuations are identical, yet in terms of the graphs that guarantee EF1 for two agents, there is no difference between identical and non-identical valuations. Intriguingly, even for more than two agents, we do not know of a graph which guarantees EF1 for identical valuations, but fails it for non-identical valuations.

\section{EF1 existence for three agents: A moving-knife protocol}\label{sec:three}
We will now consider the case of three agents. \citet{Stromquist1980} designed a protocol that results in an envy-free contiguous allocation of a divisible cake. We now give a brief outline of the protocol, illustrated by Figure~\ref{fig:stromquist}. 

\begin{figure}[ht]
	\centering
	\begin{tikzpicture}
	\draw (-0.3,0) rectangle (5,1);
	
	\begin{scope}[shift={(0.2,0)}]
	\draw (3,1.07) -- (3,-0.4);
	\draw (2.92, -0.2) rectangle (3.08, -0.9);
	\node (arr) at (3.3, -0.8) {\footnotesize$\to$};
	\end{scope}
	\begin{scope}[shift={(-0.5,0)}]
	\draw (3,0.95) -- (3,-0.4);
	\draw (2.92, -0.2) rectangle (3.08, -0.9);
	\node (arr) at (3.3, -0.8) {\footnotesize$\to$};
	\end{scope}
	\begin{scope}[shift={(1.0,0)}]
	\draw (3,0.95) -- (3,-0.4);
	\draw (2.92, -0.2) rectangle (3.08, -0.9);
	\node (arr) at (3.3, -0.8) {\footnotesize$\to$};
	\end{scope}
	
	\draw[thick] (1,1.07) -- (1,-0.4);
	\draw[fill=gray!30] (0.92, -0.3) rectangle (1.08, -0.8);
	\draw[fill=white] (0.7, -0.2) rectangle (1.3, -0.3);
	\draw (0.92, -0.3) rectangle (1.08, -0.9);
	
	\node (arr1) at (1.3, -0.8) {\footnotesize$\to$};
	
	\draw [decorate,decoration={brace,amplitude=4pt},yshift=0pt](-0.3,1.1) -- (0.98,1.1) node [black,midway,yshift=10pt] {$L$};
	\draw [decorate,decoration={brace,amplitude=4pt},yshift=0pt](1.02,1.1) -- (3.18,1.1) node [black,midway,yshift=10pt] {$M$};
	\draw [decorate,decoration={brace,amplitude=4pt},yshift=0pt](3.22,1.1) -- (5,1.1) node [black,midway,yshift=10pt] {$R$};
	
	\end{tikzpicture}
	\caption{Stromquist's moving knife protocol}
	\label{fig:stromquist}
\end{figure}
A referee holds a sword over the cake. Each of the three agents holds their own knife over the portion of the cake to the right of the sword, positioning it so that this portion is divided into two pieces they judge to have the same value. Now, initially, the sword is at the left end of the cake. It starts moving at a constant speed from left to right, while the agents continuously move their knives to keep dividing the right-hand portion into equally-valued pieces. At some point (when the leftmost piece becomes valuable enough), one of the agents shouts ``cut'', and the cake will be cut twice: once by the sword, and once by the middle one of the three knives. Agents shout ``cut'' as soon as the left piece is a highest-valued piece among the three. The agent who shouts receives the left piece. The remaining agents each receive a piece containing their knife. The resulting allocation is envy-free, since the agent receiving the left piece prefers it to the other pieces, and the other agents who are not shouting receive at least half the value of the part of the cake to the right of the sword.

Let $G$ be a path, $P=(v_1,v_2,\ldots,v_m)$.
There are several difficulties in translating Stromquist's continuous procedure to the discrete setting for $G$.
First, agents need to divide the piece to the right of the sword in half, and this might not be possible exactly given indivisibilities; but this can be handled using our concept of lumpy ties from Section~\ref{sec:two}.
Next, when the sword moves one item to the right, the lumpy ties of the agents may need to jump several items to the right, for example, because the new member of the leftmost bundle is very valuable.
To ensure EF1, we will need to smoothen these jumps, so that the middle piece grows one item at a time.
Also, it will be helpful to have the sword move in half-steps: it alternates between being placed between items (so it cuts the edge between the items), and being placed over an item, in which case the sword covers the item and agents ignore that item.
Finally, while the sword covers an item, we will only terminate if at least \emph{two} agents shout to indicate that they prefer the leftmost piece; this will ensure that there is an agent who is flexible about which of the bundles they are assigned.
The algorithm moves in steps, and alternates between moving the sword, and updating the lumpy ties. 

In our formal description of the algorithm, we do not use swords and knives. Instead, we maintain three bundles $L$, $M$, and $R$ that can be seen as resulting from a certain configuration this cutting implements.
We also need a few definitions.
For a subsequence of vertices $P(v_s,v_r)=(v_s,v_{s+1},\ldots,v_r)$ and an agent $i$, recall that $v_j$ ($s\le j \le r$) is the \emph{lumpy tie} over $P(v_s,v_r)$ for $i$ if $j$ is the smallest index such that
\begin{equation}
\label{eq:def-lumpy-tie}
u_i(L(v_j) \cup \{ v_j \}) \ge u_i(R(v_j)) \quad \text{and}\quad u_i(R(v_j) \cup \{ v_j \}) \ge u_i(L(v_j)).
\end{equation}
Here, the definitions of $L(v_j)$ and $R(v_j)$ apply to the subsequence $P(v_s,v_r)$.
The lumpy tie always exists by the discussion after equation \eqref{eq:def-lumpy-tie-1}.
Each of the three agents has a lumpy tie over $P(v_s,v_r)$; a key concept for us is the \emph{median lumpy tie} which is the median of the lumpy ties of the three agents, where the median is taken with respect to the ordering of $P(v_s,v_r)$.
We say that $i \in N$ is a \emph{left agent} (respectively, a \emph{middle agent} or a \emph{right agent}) over $P(v_s,v_r)$ if the lumpy tie for $i$ appears strictly before (respectively, is equal to, or appears strictly after) the median lumpy tie.
Note that by definition of the median, there is at most one left agent, at most one right agent, and at least one middle agent. Suppose that the median lumpy tie over the subsequence $P(v_s,v_r)$ is $v_j$, and let $i$ be an agent. Then using the definitions of lumpy tie and left/right agents, we find that
\begin{equation}
\label{eq:left-right-lumpy-tie}
\begin{array}{l}
u_i(L(v_j)) \ge u_i(R(v_j) \cup \{ v_j \}) \quad\text{if $i$ is a left agent, and } \\
u_i(R(v_j)) \ge u_i(L(v_j) \cup \{ v_j \}) \quad\text{if $i$ is a right agent.}
\end{array}
\end{equation}
Given the median lumpy tie $v_j$ over $P(v_s,v_r)$, and a two-agent set $S=\{i,k\} \subseteq N$, we define $\Lumpy(S,v_j,P(v_s,v_r))$ to be the allocation of the items in $P(v_s,v_r)$ to $S$ such that
\begin{itemize}
	\item if $i$ is a left agent and $k$ is a right agent, then $i$ receives $L(v_j)$ and $k$ receives $R(v_j) \cup \{ v_j \}$;
	\item if $i$ is a middle agent, then agent $k$ receives $k$'s preferred bundle among $L(v_j)$ and $R(v_j)$, and agent $i$ receives the other bundle along with $v_j$.
\end{itemize}
Using \eqref{eq:def-lumpy-tie} and \eqref{eq:left-right-lumpy-tie}, we see that $\Lumpy(S,v_j,P(v_s,v_r))$ is an EF1 allocation:

\begin{lemma}[Median Lumpy Ties Lemma]\label{lem:lumpy}
	 Let $S=\{i,k\} \subseteq N$ and let $v_j$ be the median lumpy tie over $P(v_s,v_r)$. Then $\Lumpy(S, v_j,P(v_s,v_r))$ is an EF1 allocation of the items in $P(v_s,v_r)$ to $S$. Further, each agent in $S$ weakly prefers their bundle to $L(v_j)$ and $R(v_j)$.
\end{lemma}

The algorithm is specified in Definition~\ref{def:stromquist-algo}. 
It alternately moves a left pointer $\ell$ (in Steps 2 and 3) and a right pointer $r$ (in Step 4). It also maintains bundles $L$, $M$, and $R$ during the execution of the algorithm. 

\begin{definition}
	\label{def:stromquist-algo}
	The \textbf{discrete moving-knife protocol for $n=3$ agents} on a sequence $P=(v_1,v_2,\ldots,v_m)$ proceeds as follows. 
	We say that an agent $i\in N$ is a \emph{shouter} if $u_i(L) \ge u_i(M)$ and $u_i(L) \ge u_i(R)$.
	
	\begin{itemize}
		\item \emph{Step 1.} Initialize $\ell = 0$ and set $r$ so that $v_r$ is the median lumpy tie over the subsequence $P(v_2,v_m)$. Initialize $L=\emptyset$, $M=\{v_{2},v_{3},\ldots,v_{r-1}\}$, and $R=\{v_{r+1},v_{r+2},\ldots,v_{m}\}$.
		\item \emph{Step 2.} Add an additional item to $L$, i.e., set $\ell=\ell+1$ and $L=\{v_1,v_2,\ldots,v_\ell\}$. 
		If no agent shouts, go to Step 3. If some agent $\shoutleft$ shouts, $\shoutleft$ receives the left bundle $L$. Allocate the remaining items according to $\Lumpy(N\setminus \{\shoutleft\},v_r,P(v_{\ell+1},v_m))$.
		\item \emph{Step 3.} Delete the leftmost point of the middle bundle, i.e., set $M=\{v_{\ell+2},v_{\ell+3},\ldots,v_{r-1}\}$.
		If the number of shouters is smaller than two, go to Step 4.
		If at least two agents shout, we show (next page) that there is a shouter $\shouter$ who is a middle agent over $P(v_{\ell+1},v_m)$. Then, allocate $L$ to a shouter $\shoutleft$ distinct from $\shouter$. Let the agent $\chooser$ distinct from $\shouter$ and $\shoutleft$ choose his preferred bundle among $\{v_{\ell+1}\}\cup M$ and $\{v_r\}\cup R$. Agent $\shouter$ receives the other bundle.
		\item \emph{Step 4.}
		If $v_r$ is the median lumpy tie over $P(v_{\ell+2},v_m)$, directly move to the following cases (a)--(d).
		If $v_r$ is not the median lumpy tie over $P(v_{\ell+2},v_m)$, set $r=r+1$, $M=\{v_{\ell+2},v_{\ell+3},\ldots,v_{r-1}\}$, and $R=\{v_{r+1},v_{r+2},\ldots,v_{m}\}$; then, go to cases (a)--(d).
		\begin{enumerate}
			\item[(a)] If at least two agents shout, find a shouter $\shouter$ who did not shout at the previous step. If there is a shouter $\shoutleft$ who shouted at the previous step, $\shoutleft$ receives $L$; else, give $L$ to an arbitrary shouter $\shoutleft$ distinct from $\shouter$. The agent $\chooser$ distinct from $\shouter$ and $\shoutleft$ choose his preferred bundle among $\{v_{\ell+1}\}\cup M$ and $\{v_r\}\cup R$, breaking ties in favor of the former option. Agent $\shouter$ receives the other bundle.
			\item[(b)] If $v_r$ is the median lumpy tie over $P(v_{\ell+2},v_m)$ and only one agent $\shoutleft$ shouts, give $L \cup \{v_{\ell + 1} \}$ to $\shoutleft$ and allocate the rest according to $\Lumpy(N\setminus \{\shoutleft\},v_r,P(v_{\ell+2},v_m))$.
			\item[(c)] If $v_r$ is the median lumpy tie over $P(v_{\ell+2},v_m)$ but no agent shouts, go to Step 2.
			\item[(d)] Otherwise $v_r$ is not the median lumpy tie over $P(v_{\ell+2},v_m)$: Repeat Step 4.
		\end{enumerate}
	\end{itemize}
\end{definition}

\begin{theorem}
	\label{thm:EF1-3agents}
The moving-knife protocol finds an EF1 allocation for three agents and runs in $O(m)$ time, when $G$ is a path. 
\end{theorem}
\begin{proof}
The algorithm is well-defined -- there is one place where this is not immediate: If two agents shout in Step 3, the algorithm description claims that there is a shouter who is a middle agent over the subsequence $P(v_{\ell+1},v_m)$. Suppose for the moment that there is a shouter $i$ who is a \emph{right} agent. Due to \eqref{eq:left-right-lumpy-tie}, we have $u_i(R) \ge u_i(\{v_{\ell+1}\} \cup M \cup \{v_r\})$. Since $i$ is a shouter, we have $u_i(L) \ge u_i(R)$, so $u_i(L) \ge u_i(\{v_{\ell+1}\} \cup M \cup \{v_r\})$. But $i$ did not shout in the previous Step 2 (when no-one shouted), so either $u_i(R) > u_i(L)$ or $u_i(\{v_{\ell+1}\} \cup M)  > u_i(L)$, and either case is a contradiction. Hence neither of the at least two shouters of Step 3 is a right agent, so at least one shouter is a middle agent, since there is at most one left agent.
	
The algorithm terminates and returns an allocation, since the bundle $L$ grows throughout the algorithm until eventually, at least two agents will think that $L$ is a best bundle and thus will shout and thereby terminate the algorithm. We will now consider every possible way that the algorithm could have terminated, and show that the resulting allocation is EF1. 
\smallskip

\noindent
\textbf{Step 2.}
\begin{itemize}[leftmargin=18pt]
	\item Agent $\shoutleft$ receives $L$ and does not envy the other agents (up to good $v_r$) since $\shoutleft$ is a shouter.
	\item An agent $i$ who is not a shouter does not envy $\shoutleft$ because $i$ prefers either $M$ or $R$ to $L$, and hence by Lemma \ref{lem:lumpy} receives a bundle preferred to $L$.
	\\ \smallskip
	Agent $i$ also does not envy the other agent $j \neq \shoutleft$ up to one good by Lemma \ref{lem:lumpy}.
	\item An agent $i \neq \shoutleft$ who is a shouter does not envy $\shoutleft$ up to one good: If this is the first time Step 2 was performed, then $L = \{v_1\}$, so $i$ does not envy $\shoutleft$ up to $v_1$. Otherwise, the last step was an iteration of Step 4(c), where by definition of Step 4(c) no-one shouted. Since $i$ did not shout during Step 4(c), and Step 2 did not change the bundles $M$ and $R$, then $i$ strictly prefers either $M$ or $R$ to the left bundle $L\setminus \{v_{\ell}\}$ of Step 4(c). By Lemma \ref{lem:lumpy}, agent $i$ gets a bundle at least as good as $M$ or $R$. Thus, $i$ does not envy $\shoutleft$ up to $v_\ell$.
	\\ \smallskip
	Also by Lemma \ref{lem:lumpy}, agent $i$ does not envy the other agent $j \neq \shoutleft$ up to one good.
\end{itemize}

\noindent
\textbf{Step 3.}
\begin{itemize}[leftmargin=18pt]
	\item Agent $\shoutleft$ receives $L$ and, because $\shoutleft$ shouted, does not envy the bundle $\{v_{\ell+1}\} \cup M$ up to good $v_{\ell+1}$, and does not envy the bundle $\{v_r\} \cup R$ up to good $v_r$.
	\item Agent $\chooser$ gets his preferred bundle among $\{v_{\ell+1}\} \cup M$ and $\{v_r\} \cup R$, and so does not envy agent $\shouter$ who receives the other bundle. Further, agent $\chooser$ does not envy agent $\shoutleft$ since $\chooser$ did not shout at the last Step 2 (where no-one shouted), which, since bundle $L$ did not change in Step 3, means that $\chooser$ prefers either $\{v_{\ell+1}\} \cup M$ or $R$ to $L$, and hence also prefers his chosen bundle to $L$.
	\item Agent $\shouter$ is a middle agent, so the lumpy tie of $\shouter$ over $P(v_{\ell+1},v_m)$ is $v_r$, and hence by \eqref{eq:def-lumpy-tie},
	\begin{equation}
		u_{\shouter}(\{v_r\} \cup R) \ge u_{\shouter}(\{v_{\ell+1}\} \cup M). \label{eq:step3:lumpy-tie}
	\end{equation}
	Now, agent $\shouter$ did not shout at the preceding Step 2 (when no-one shouted). However, $\shouter$ \emph{does} shout after deleting $v_{\ell+1}$ from $M$. Since $L$ and $R$ have not changed, the reason $\shouter$ did not shout at Step 2 was that $L$ is worse than the middle bundle during Step 2, so
	\begin{equation}
		u_{\shouter}(\{v_{\ell+1}\} \cup M) > u_{\shouter}(L). \label{eq:step3:M-better-L}
	\end{equation}
	Combining \eqref{eq:step3:lumpy-tie} and \eqref{eq:step3:M-better-L}, we also have
	\[
	u_{\shouter}(\{v_r\} \cup R) > u_{\shouter}(L).
	\]
	 Since $\shouter$ receives either $\{v_{\ell+1}\} \cup M$ or $\{v_r\} \cup R$, agent $\shouter$ does not envy agent $\shoutleft$ receiving $L$.
	
	Finally, from \eqref{eq:step3:lumpy-tie}, agent $\shouter$ weakly prefers $\{v_r\} \cup R$ to $\{v_{\ell+1}\} \cup M$. Thus, if $\chooser$ picks $\{v_{\ell+1}\} \cup M$, then $\shouter$ does not envy $\chooser$. On the other hand, if $\chooser$ picks the bundle $\{v_r\} \cup R$, then $\shouter$ does not envy $\chooser$ up to good $v_r$: we have $u_\shouter(L) \ge u_\shouter(R)$ since $\shouter$ shouts, and so by \eqref{eq:step3:M-better-L}, also
	\[ u_{\shouter}(\{v_{\ell+1}\} \cup M) > u_{\shouter}(R). \]
\end{itemize}

\noindent
\textbf{Step 4(a).}
We first prove that if $i$ is a shouter who did not shout in the previous step, then
\begin{align}\label{eq:step4:new-shouter}
u_i(\{v_r\}\cup R) > u_i(L) \ge u_i(M).
\end{align}
In the previous step (which was either Step 3 or Step 4), the middle bundle was $M \setminus \{v_{r-1}\}$ and the right bundle was $\{v_r\} \cup R$. (While Step 4 allows for the possibility that the middle and right bundles are not changed in Step 4, this is not the case if we enter Step 4(a): if the bundles are unchanged and two agents shout, these agents already shouted in Step 3, contradicting that we did not terminate then.)
Since $i$ did not shout with the middle and right bundles of the previous step, we have
 \[
u_i(M \setminus \{v_{r-1}\}) > u_i(L) \quad\text{or}\quad u_i(\{v_r\}\cup R) > u_i(L).
\]
Since $i$ is a shouter, $u_i(L) \ge u_i(M)$, so that the first case is impossible by monotonicity. Hence $u_i(\{v_r\}\cup R) > u_i(L)$, showing \eqref{eq:step4:new-shouter}, when combined with $u_i(L) \ge u_i(M)$.

\begin{itemize}[leftmargin=18pt]
	\item Agent $\shoutleft$ receives $L$ and does not envy other agents up to one good like in Step 3.
	\item Agent $\chooser$ gets his preferred bundle among $\{v_{\ell+1}\} \cup M$ and $\{v_r\} \cup R$, and so does not envy agent $\shouter$ who receives the other bundle. Agent $\chooser$ also does not envy $\shoutleft$: If $\chooser$ is not a shouter, then $\chooser$ does not envy $\shoutleft$ because $\chooser$ prefers either $M$ or $R$ to $L$, and hence prefers his picked piece to $L$. If $\chooser$ is a shouter, then all three agents are shouters, and by choice of $\chooser$, this means that $\chooser$ was not a shouter at the previous step, when there was at most one shouter. By \eqref{eq:step4:new-shouter}, $u_{\chooser}(\{v_r\}\cup R) > u_{\chooser}(L)$, and hence
	\[
	\max \{u_{\chooser}(\{v_{\ell+1}\}\cup M),u_{\chooser}(\{v_r\}\cup R)\} \geq u_{\chooser}(L),
	\]
	so that $\chooser$ does not envy $\shoutleft$.
	\item Agent $\shouter$ does not envy others up to one good:
	\begin{itemize}
		\item Suppose agent $\chooser$ strictly prefers $\{v_r\}\cup R$ to $\{v_{\ell+1}\}\cup M$. Then agent $\chooser$'s lumpy tie over $P(v_{\ell+1},v_m)$ appears at or after $v_r$ by definition of the lumpy tie. As we argued before, the bundles $M$ and $R$ were changed in the execution of Step 4, and $r$ was increased by 1. Thus, $v_r$ appears strictly after the median lumpy tie over $P(v_{\ell+1},v_m)$. Thus, $\chooser$ is the right agent over $P(v_{\ell+1},v_m)$. Hence $\shouter$ is either a left or middle agent over $P(v_{\ell+1},v_m)$ since there is at most one right agent. Using \eqref{eq:def-lumpy-tie} or \eqref{eq:left-right-lumpy-tie}, this implies
		\begin{equation}
			u_\shouter(\{v_{\ell+1}\}\cup M) \ge u_\shouter(\{v_r\}\cup R),
			\label{eq:step4:shouter-lumpy}
		\end{equation}
		so that $\shouter$ does not envy $\chooser$.
		
		By definition of $\shouter$, agent $\shouter$ did not shout in the previous step. By \eqref{eq:step4:new-shouter}, $u_{\shouter}(\{v_r\}\cup R) \ge u_{\shouter}(L)$, so together with \eqref{eq:step4:shouter-lumpy}, we have $u_\shouter(\{v_{\ell+1}\}\cup M) \ge u_{\shouter}(L)$, so $\shouter$ does not envy $\shoutleft$.
		\item Suppose $\chooser$ weakly prefers $\{v_{\ell+1}\}\cup M$ to $\{v_r\}\cup R$. Then $\shouter$ receives the bundle $\{v_r\}\cup R$ (since $\chooser$ breaks ties in favor of $\{v_{\ell+1}\}\cup M$). By choice of $\shouter$, agent $\shouter$ did not shout at the last step. So by \eqref{eq:step4:new-shouter}, we have $u_\shouter(\{v_r\}\cup R) > u_\shouter(L)$ so that $\shouter$ does not envy $\shoutleft$, and also by \eqref{eq:step4:new-shouter}, we have $u_\shouter(\{v_r\}\cup R) > u_\shouter(M)$ so that $\shouter$ does not envy $c$ up to item $v_{\ell+1}$.
	\end{itemize}
\end{itemize}

\textbf{Step 4(b).}
\begin{itemize}
	\item Agent $\shoutleft$ gets $L \cup \{v_{\ell+1}\}$ and does not envy the other agents (up to good $v_r$) as $\shoutleft$ shouts.
	\item Any agent $i \neq \shoutleft$ is not a shouter, and thus prefers either $M$ or $R$ to $L$. Hence by Lemma \ref{lem:lumpy} receives a bundle preferred to $L$, and so does not envy $\shoutleft$ up to item $v_{\ell+1}$. \\ \smallskip
	Agent $i$ also does not envy the other agent $j \neq \shoutleft$ up to one good by Lemma \ref{lem:lumpy}.
\end{itemize}

Thus, the allocation returned by any of the steps satisfies EF1. Our algorithm can be implemented in $O(m)$ time:
Each of steps 2, 3, and 4 will be executed at most $m$ times (since $\ell$ and $r$ can only be incremented $m$ times).
The execution of each step takes constant time: In each step, we need to check which agents shout, and this can be done in a constant number of queries to agents' valuations; also, in Step 4 we need to calculate the lumpy ties of the agents, but this can be done in amortized constant time, since during the execution of the algorithm, the position of each agent's lumpy tie can only move to the right. 
Finally, when enough agents shout, we can clearly compute and return the final allocation in $O(m)$ time.
\end{proof}

\section{EF2 existence for any number of agents}\label{sec:EF2}
For two or three agents, we have seen algorithms that are guaranteed to find an EF1 allocation on a path (and on traceable graphs). Both algorithms were adaptations of procedures that identify envy-free divisions in the cake-cutting problem. For the case of four or more agents, we face a problem: there are no known procedures that find connected envy-free division in cake-cutting if the number of agents is larger than three. However, in the divisible setting, a non-constructive existence result is known: \citet{Su1999} proved, using Sperner's lemma, that for any number of agents, a connected envy-free division of a cake always exists. One might try to use this result as a black box to obtain a fair allocation for the indivisible problem on a path: Translate an indivisible instance with additive valuations into a divisible cake (where each item corresponds to a region of the cake), obtain an envy-free division of the cake, and round it to get an allocation of the items. \citet{Suksompong2017} followed this approach and showed that the result is an allocation where any agent $i$'s envy $u_i(A(j)) - u_i(A(i))$ is at most $2u_{\text{max}}$, where $u_{\text{max}}$ is the maximum valuation for a single item.

In this section, rather than using \citeauthor{Su1999}'s~\citeyearpar{Su1999} result as a black box, we directly apply Sperner's lemma to the indivisible problem. This allows us to  obtain a stronger fairness guarantee: We show that on paths (and on traceable graphs), there always exists an EF2 allocation.%
\footnote{\label{footnote:weak-ef1}To see that EF2 is a stronger property than bounding envy up to $2u_{\text{max}}$, consider a path of four items and two agents with additive valuations $1$--$10$--$2$--$2$. The allocation $(1,10$--$2$--$2)$ is not EF2, but the first agent has an envy of $13 < 20 = 2u_{\text{max}}$.}
An allocation is EF2 if any agent's envy can be avoided by removing up to two items from the envied bundle. Again, we only allow removal of items if this operation leaves a connected bundle. 
\begin{definition}[EF2: envy-freeness up to two outer goods]
	An allocation $A$ satisfies EF2 if, for any pair $i,j \in N$ of agents, either $|A(j)| \le 1$, or there are two goods $u, v \in A(j)$ such that $A(j) \setminus \{u,v\}$ is connected and $u_i(A(i)) \ge u_i(A(j)\setminus \{u,v\})$.
\end{definition}

Let us first give a high-level illustration with three agents of how Sperner's lemma can be used to find low-envy allocations.  

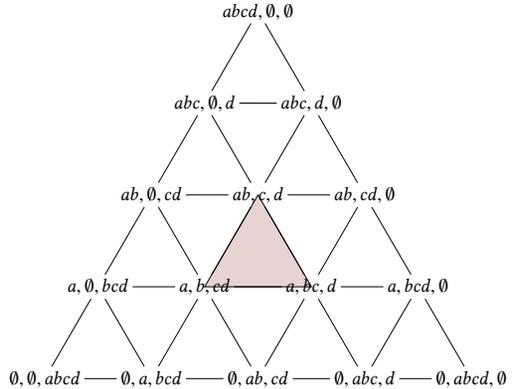
\begin{figure}[ht]
	\centering
	\scalebox{1}{ 
		\begin{tikzpicture}[xscale=2, yscale=1.7, every node/.style={fill=none, inner sep=0.2em}]
		
		\draw[fill=black!10!white!92!red] (1.5,1) -- (2.5,1) -- (2,2) -- cycle;
		
		\node (11abcd) at (0,0) {$\emptyset,\emptyset,abcd$};
		\node (1a1bcd) at (1,0) {$\emptyset,a,bcd$};
		\node (1ab1cd) at (2,0) {$\emptyset,ab,cd$};
		\node (1abc1d) at (3,0) {$\emptyset,abc,d$};
		\node (1abcd1) at (4,0) {$\emptyset,abcd,\emptyset$};
		
		\node (a11bcd) at (0.5,1) {$a,\emptyset,bcd$};
		\node (a1b1cd) at (1.5,1) {$a,b,cd$};
		\node (a1bc1d) at (2.5,1) {$a,bc,d$};
		\node (a1bcd1) at (3.5,1) {$a,bcd,\emptyset$};
		
		\node (ab11cd) at (1.0,2) {$ab,\emptyset,cd$};
		\node (ab1c1d) at (2.0,2) {$ab,c,d$};
		\node (ab1cd1) at (3.0,2) {$ab,cd,\emptyset$};
		
		\node (abc11d) at (1.5,3) {$abc,\emptyset,d$};
		\node (abc1d1) at (2.5,3) {$abc,d,\emptyset$};
		
		\node (abcd11) at (2.0,4) {$abcd,\emptyset,\emptyset$};
		
		\draw
		(11abcd) -- (1a1bcd)
		(1a1bcd) -- (1ab1cd)
		(1ab1cd) -- (1abc1d)
		(1abc1d) -- (1abcd1)
		
		(a11bcd) -- (a1b1cd)
		(a1b1cd) -- (a1bc1d)
		(a1bc1d) -- (a1bcd1)
		
		(ab11cd) -- (ab1c1d)
		(ab1c1d) -- (ab1cd1)
		
		(abc11d) -- (abc1d1)

		(11abcd) -- (a11bcd)
		(a11bcd) -- (ab11cd)
		(ab11cd) -- (abc11d)
		(abc11d) -- (abcd11)
		
		(1a1bcd) -- (a1b1cd)
		(a1b1cd) -- (ab1c1d)
		(ab1c1d) -- (abc1d1)
		
		(1ab1cd) -- (a1bc1d)
		(a1bc1d) -- (ab1cd1)
		
		(1abc1d) -- (a1bcd1)
		
		(abcd11) -- (abc1d1)
		(abc1d1) -- (ab1cd1)
		(ab1cd1) -- (a1bcd1)
		(a1bcd1) -- (1abcd1)
		
		(abc11d) -- (ab1c1d)
		(ab1c1d) -- (a1bc1d)
		(a1bc1d) -- (1abc1d)
		
		(ab11cd) -- (a1b1cd)
		(a1b1cd) -- (1ab1cd)
		
		(a11bcd) -- (1a1bcd)
		;
		
		\end{tikzpicture}}
	\caption{Connected partitions as a subdivided simplex}
	\label{fig:simplex}
\end{figure}
Given a path $P = (a,b,c,d)$, the family of connected partitions of $P$ can naturally be arranged as the vertices of a subdivided simplex, as in Figure~\ref{fig:simplex}.

For each of these partitions, each agent $i$ \textit{labels} the corresponding vertex by the index of a bundle from that partition that $i$ most-prefers. For example, the top vertex will be labelled as ``index 1'' by all agents, since they all most-prefer the leftmost bundle in $(abcd,\emptyset,\emptyset)$. Now, Sperner's lemma will imply that at least one of the simplices (say the shaded one) is ``fully-labeled'', which means that the first agent most-prefers the leftmost bundle at one vertex, the second agent most-prefers the middle bundle at another vertex, and the third agent most-prefers the rightmost bundle at the last vertex. Notice that the partitions at the corner points of the shaded simplex are all ``similar'' to each other (they can be obtained from each other by moving only one item). Hence, we can ``round'' the corner-partitions into a common allocation $A^*$, say by picking one of the corner partitions arbitrarily and then allocating bundles to agents according to the labels. The resulting allocation has the property that any agents' envy can be eliminated by moving at most one good.%
\footnote{One can generalize this argument to show that on paths, there exists an allocation $A$ satisfying a weak form of EF1: for any $i,j\in [n]$, we have $u_i(I_i \cup \{g_i\}) \ge u_i(I_j \setminus \{g_j\})$ for some items $g_i,g_j$ such that $I_i \cup \{g_i\}$ and $I_j \setminus \{g_j\}$ are connected. For additive valuations, this implies that envy is bounded by $u_i(g_i) + u_i(g_j) \le 2u_{\text{max}}$, which is the result of \citet{Suksompong2017}.}

The argument sketched above does not yield an EF1 nor even an EF2 allocation. Intuitively, the problem is that the connected partitions at the corners of the fully-labeled simplex are ``too far apart'', so that no matter how we round the corner partitions into a common allocation $A^*$, some agents' bundles will have changed too much, and so we cannot prevent envy even up to one or two goods. In the following, we present a solution to this problem, by considering a finer subdivision: we introduce $n-1$ knives which move in half-steps (rather than full steps), and which might `cover' an item so that it appears in none of the bundles. The result is that the partial partitions in the corners of the fully-labeled simplex are closer together, and can be successfully rounded into an EF2 allocation $A^*$.

In our approach, we use a specific triangulation (Kuhn's triangulation, \citealp{Kuhn1960}). This triangulation has the needed property that the partitions at the corners of sub-simplices are close together, and adjacent partitions can be obtained from each other in a natural way. While this type of triangulation has also been used in cake-cutting, e.g., by \citet{Deng2012}, there it was only used to speed up algorithms (compared to the barycentric subdivision used by \citet{Su1999}), not to obtain better fairness properties.

\subsection{Sperner's lemma}
We start by formally introducing Sperner's lemma \citep[cf.][]{Flegg1974}. Let $\conv(\vv_1,\vv_2,\ldots, \vv_k)$ denote the convex hull of $k$ vectors $\vv_1,\vv_2,\ldots, \vv_k$. An \emph{$n$-simplex} is an $n$-dimensional polytope which is the convex hull of its $n+1$ \emph{main vertices}. A \emph{$k$-face} of the $n$-simplex is the $k$-simplex formed by the span of any subset of $k+1$ main vertices. A \emph{triangulation} $T$ of a simplex $S$ is a collection of sub-$n$-simplices whose union is $S$ with the property that the intersection of any two of them is either the empty set, or a face common to both. Each of the sub-simplices $S^* \in T$ is called an \emph{elementary} simplex of the triangulation $T$. We denote by $V(T)$ the set of vertices of the triangulation $T$, i.e., the union of vertices of the elementary simplices of $T$.

Let $T$ be some fixed triangulation of an $(n-1)$-simplex $S=\conv(\vv_1,\vv_2,\ldots,$ $\vv_n)$.
A \emph{labeling function} is a function $L : V(T) \rightarrow [n]$ that assigns a number in $[n]$ (called a \emph{color}) to each vertex of the triangulation $T$.
A labeling function $L$ is called \emph{proper} if
\begin{itemize}
	\item For each main vertex $\vv_i$ of the simplex, $L$ assigns color $i$ to $\vv_i$: $L(\vv_i)=i$; and
	\item $L(\vv)\neq i$ for any vertex $\vv \in V(T)$ belonging to the $(n-2)$-face of $S$ not containing $\vv_i$.
\end{itemize}
Sperner's lemma states that if $L$ is a proper labeling function, then there exists an elementary simplex of $T$ whose vertices have all different labels.

We will consider a generalized version of Sperner's lemma, proved, for example, by \citet{Bapat1989}. In this version, there are $n$ labeling functions $L_1,\dots,L_n$, and we are looking for an elementary simplex that is \emph{fully-labeled} for some way of assigning labeling functions to vertices, where we must use each labeling function exactly once. The formal definition is as follows.

\begin{definition}[Fully-labeled simplex]
	Let $T$ be a triangulation of an $(n-1)$-simplex, and let $L_1,\dots,L_n$, be labeling functions. An elementary simplex $S^*$ of $T$ is \emph{fully-labeled} if we can write $S^*= \conv(\vv^*_1,\vv^*_2,\ldots, \vv^*_n)$ such that there exists a permutation $\phi:[n] \rightarrow [n]$ with
	\[ L_{i}(\vv^*_{i})=\phi(i) \quad\text{for each $i\in [n]$}. \]
\end{definition}

The generalized version of Sperner's lemma that we consider, taken from \citet{Bapat1989}, guarantees the existence of a fully-labeled simplex.

\begin{lemma}[Generalized Sperner's Lemma]
	\label{lem:sperner}
	Let $T$ be a triangulation of an $(n-1)$-simplex~$S$, and let $L_1,\dots,L_n$ be proper labeling functions.
	Then there is a fully-labeled simplex $S^*$ of $T$.
\end{lemma}

\subsection{Existence of EF2 allocations}
Suppose that our graph $G$ is a path $P=(1,2,\ldots, m)$, where the items are named by integers. We assume that $m\ge n$, so that there are at least as many items as agents (when $m < n$ it is easy to find EF1 allocations). Our aim is to cut the path $P$ into $n$ intervals (bundles) $I_*^{1},I_*^{2},\ldots,I_*^{n}$. Throughout the argument, we use superscripts to denote indices of bundles; index 1 refers to the leftmost bundle and index $n$ refers to the rightmost bundle.

\paragraph{Construction of the triangulation.}
Consider the $(n-1)$-simplex\footnote{The simplex $S_m$ is  affinely equivalent to the standard $(n-1)$-simplex $\Delta_{n-1} = \{ (l_1,\dots,l_n) \ge 0 : \sum l_i = 1 \}$ via $x_i = m \cdot (l_1 + l_2 + \cdots + l_i) + \frac12$. In these coordinates, $l_i$ is the length of the $i$-th piece (times $1/m$).}
\begin{align}\label{eq:simplex}
\textstyle
S_m=\{\, \bfx \in \mathbb R^{n-1} : \frac12 \leq x^1\leq x^2\leq \ldots\leq  x^{n-1}\leq m + \frac12 \,\}.
\end{align}
We construct a triangulation $T_{\half}$ of $S_m$ whose vertices $V(T_{\half})$ are the points $\bfx \in S_m$ such that each $x^j$ is either integral or half-integral, namely,
\[
\textstyle
V(T_{\half}) = \{ \bfx \in S_m : x^j \in \{ \frac12, 1, \frac32, 2, \frac52 \dots, m, m + \frac12 \} \text{ for all $j\in [n]$}   \}.
\]
For reasons that will become clear shortly, we call a vector $\bfx \in V(T_{\half})$ a \emph{knife position}.

Using Kuhn's triangulation \citep{Kuhn1960,Scarf1982,Deng2012}, we  construct $T_{\half}$ so we can write each elementary simplex $S' \in T_{\half}$ as $S' = \conv(\bfx_1,\bfx_2,\ldots \bfx_n)$ and there is a permutation $\pi:[n]\rightarrow [n]$ with
\begin{equation}
\label{eq:sperner:kuhn-triangulation}
\bfx_{i+1}=\bfx_i+\textstyle \frac12 \mathbf{e}^{\pi(i)} \quad \text{for each $i\in [n-1]$,}
\end{equation}
where $\mathbf{e}^j = (0,\dots,1,\dots,0)$ is the $j$-th unit vector. 

We give an interpretation of \eqref{eq:sperner:kuhn-triangulation} shortly. Each vertex $\bfx=(x^1,x^2,\ldots,x^{n-1}) \in V(T_{\half})$ of the triangulation $T_{\half}$ corresponds to a partial partition $A({\bfx})=(I^1(\bfx),I^2(\bfx),\ldots,I^n(\bfx))$ of $P$ where
$I^{j}(\bfx) := \{ y \in \{ 1, 2, \dots, m \}  : x^{j-1} < y < x^j \}$,
writing $x^0 = \frac12$ and $x^n = m+\frac12$ for convenience.
Intuitively, $\bfx$ specifies the location of $n-1$ knives that cut $P$ into $n$ pieces. If $x^j$ is integral, that is $x^j \in \{ 1,\dots,m \}$, then the $j$-th knife `covers' the item $x^j$, which is then part of neither $I^{j}(\bfx)$ nor $I^{j+1}(\bfx)$. This is why $A({\bfx})$ is a \emph{partial} partition. Since there are only $n-1$ knives but $m \ge n$ items, not all items are covered, so at least one bundle is non-empty.

Property \eqref{eq:sperner:kuhn-triangulation} means that, if we visit the knife positions $\bfx_1,\bfx_2,\ldots \bfx_n$ at the corners of an elementary simplex in the listed order, then at each step exactly one of the knives moves by half a step, and each knife moves only at one of the steps.

\paragraph{Construction of the labeling functions.}
We now construct, for each agent $i\in [n]$, a labeling function $L_i : V(T_\half) \to [n]$. The function $L_i$ takes as input a vertex $\bfx$ of the triangulation $T_{\half}$ (interpreted as the partial partition $A(\bfx)$), and returns a color in $[n]$. The color will specify the index of a bundle in $A(\bfx)$ that agent $i$ likes the most.
Formally,
\[ L_i(\bfx) \in \{ j\in [n] : u_i(I^j(\bfx)) \ge u_i(I^k(\bfx)) \text{ for all $k\in [n]$} \}. \]
If there are several most-preferred bundles in $A(\bfx)$, ties can be broken arbitrarily. However, we insist that the index $L_i(\bfx)$ always corresponds to a non-empty bundle (this can be ensured since $A(\bfx)$ always contains a non-empty bundle, and $u_i$ is monotonic).

The labeling functions $L_i$ are proper. For each $j\in [m]$, the main vertex $\vv_j$ of the simplex $S_m$ has the form $\vv_j = (\frac12,\dots,\frac12,m+\frac12,\dots,m+\frac12)$, where the first $j-1$ entries are $\frac12$ and the rest are $m+\frac12$.
In the partition $A(\vv_j)$, the bundle $I^j(\vv_j)$ contains all the items, so is most-preferred (since $u_i$ is monotonic and by our tie-breaking), and so $L_i(\vv_j) = j$. Further, any vertex $\bfx$ belonging to the $(n-2)$-face of $S_m$ not containing $\vv_j$ satisfies $x^{j-1} = x^j$, and thus in partition $A(\bfx)$, bundle $I^j(\bfx)$ is empty, hence is \emph{not} selected, and so $L_i(\bfx) \neq j$.

By the generalized version of Sperner's lemma (Lemma~\ref{lem:sperner}), there exists an elementary simplex $S^* = \conv(\bfx_1,\bfx_2,\ldots,\bfx_n)$ of the triangulation $T_{\half}$ which is fully-labeled, so that, for some permutation $\phi:[n] \rightarrow [n]$, we have $L_{i}(\bfx_{i})=\phi(i)$ for all $i\in [n]$.

\paragraph{Translation into partial partitions.}
The fully-labeled elementary simplex $S^*$ corresponds to a sequence $(A_1,A_2,\ldots, A_n)$ of partial partitions of $P$, which we call the \emph{Sperner sequence}, where $A_i= (I_i^1, \dots, I_i^n) := A(\bfx_i)$ for each $i \in [n]$.
An example of a Sperner sequence is shown in Figure~\ref{fig:sperner}.
From the labeling, for each agent $i\in [n]$, since $L_{i}(\bfx_{i})=\phi(i)$, the bundle with index $\phi(i)$ in the partition $A_i$ is a best bundle for $i$:
\begin{equation}
\label{eq:ef2:my-bundle-is-best}
u_i (I_i^{\phi(i)}) \ge u_i (I_i^j) \quad\text{for each $j\in[n]$}.
\end{equation}

Now, for each $j \in [n]$, we define the \emph{basic bundle} $B^j := I_1^j \cap \cdots \cap I_n^j$ to be the bundle of items that appear in the $j$-th bundle of every partition in the Sperner sequence. The set of basic bundles is a partial partition. Let us analyze the items between basic bundles.

From \eqref{eq:sperner:kuhn-triangulation}, each of the $n-1$ knives moves exactly once, by half a step, while passing through the Sperner sequence $(A_1,A_2,\ldots, A_n)$. Thus, the numbers $x_1^j,\dots,x_n^j$ take on two different values, one of which is integral and the other half-integral. We write $y^j$ for the integral value (so $y^j = x_i^j$ for some $i\in [n]$), and call $y^j$ a \emph{boundary item}. The $j$-th knife covers the item $y^j$ in some, but not all, of the partial partitions in the Sperner sequence. Now, there are two cases:
\begin{enumerate}
	\item[(a)] $x_1^j = \dots = x_i^j = y^j - \frac12$ and $x_{i+1}^j = \dots = x_n^j = y^j$ for some $i\in [n]$, so that $y^j$ never occurs in the $j$-th bundle in the Sperner sequence but sometimes occurs in the $(j+1)$-th bundle, or
	\item[(b)] $x_1^j = \dots = x_i^j = y^j$ and $x_{i+1}^j = \dots = x_n^j = y^j + \frac12$ for some $i\in [n]$, so that $y^j$ sometimes occurs in the $j$-th bundle in the Sperner sequence but never occurs in the $(j+1)$-th bundle.
\end{enumerate}
Since $y^j$ is sometimes covered by a knife, it is not part of any basic bundle.
Note that
\begin{equation}
\label{eq:basic-bundle-is-basic}
B^j \subseteq I^{j}_i \subseteq \{ y^{j-1} \} \cup B^j \cup \{ y^j \} \quad \text{for every $i,j \in [n]$.}
\end{equation}

\begin{figure}[t]
	\centering
	\begin{tikzpicture}[scale=0.9, transform shape, every node/.style={minimum size=5mm, inner sep=1pt}]
	
	\def\rulery{1}
	\draw (-0.6, \rulery) -- (13.8, \rulery);
	\foreach \x/\xtext in {-0.6/\frac12, 0/1, 0.6/\frac32, 1.2/2, 1.8/\frac52, 2.4/\cdots, 12.6/\quad\cdots,13.8/m+\frac12}
	\draw[shift={(\x,\rulery)}] (0pt,2.5pt) -- (0pt,-2.5pt) node[above=5.5pt] {$\strut\xtext$};
	
	\foreach \x in {3.0,3.6,4.2,4.8,5.4,6.0,6.6,7.2,7.8,8.4,9.0,9.6,10.2,10.8,11.4,12.0,12.6,13.2,13.8}
	\draw[shift={(\x,\rulery)}] (0pt,2.5pt) -- (0pt,-2.5pt);
	
	\draw[shift={(-0.6,\rulery)}] (0pt,3.5pt) -- (0pt,-3.5pt);
	\draw[shift={(13.8,\rulery)}] (0pt,3.5pt) -- (0pt,-3.5pt);
	
	\draw (-0.6,0.4) rectangle (3.0,-0.4);
	\draw[fill=black!10] (3.0,0.4) rectangle (6.6,-0.4);
	\draw (6.6,0.4) rectangle (10.8,-0.4);
	\draw (10.8,0.4) rectangle (13.8,-0.4);
	
	\node at (-1.1,0) {$A_1$};
	\node[draw, circle](1) at (0,0) {};
	\node[draw, circle](2) at (1.2,0) {};
	\node[draw, circle](3) at (2.4,0) {};
	\node[draw, circle](4) at (3.6,0) {$y_1$};
	\node[draw, circle](5) at (4.8,0) {};
	\node[draw, circle](6) at (6,0) {};
	\node[draw, circle](7) at (7.2,0) {$y_2$};
	\node[draw, circle](8) at (8.4,0) {};
	\node[draw, circle](9) at (9.6,0) {};
	\node[draw, circle](10) at (10.8,0) {$y_3$};
	\node[draw, circle](11) at (12.0,0) {};
	\node[draw, circle](12) at (13.2,0) {};
	
	\draw[-, >=latex,thick]  (1)--(2) (2)--(3) (3)--(4) (4)--(5) (5)--(6) (6)--(7) (7)--(8) (8)--(9) (9)--(10) (10)--(11) (11)--(12);
	
	\begin{scope}[shift={(0,-1.5)}]
	
	\draw[fill=black!10, draw=none] (-0.6,0.4) rectangle (3.0,-0.4);
	\draw (-0.6,0.4) rectangle (3.6,-0.4);
	\draw (3.6,0.4) rectangle (6.6,-0.4);
	\draw (6.6,0.4) rectangle (10.8,-0.4);
	\draw (10.8,0.4) rectangle (13.8,-0.4);
	
	\node at (-1.1,0) {$A_2$};
	\node[draw, circle](1) at (0,0) {};
	\node[draw, circle](2) at (1.2,0) {};
	\node[draw, circle](3) at (2.4,0) {};
	\node[draw, circle](4) at (3.6,0) {$y_1$};
	\node[draw, circle](5) at (4.8,0) {};
	\node[draw, circle](6) at (6,0) {};
	\node[draw, circle](7) at (7.2,0) {$y_2$};
	\node[draw, circle](8) at (8.4,0) {};
	\node[draw, circle](9) at (9.6,0) {};
	\node[draw, circle](10) at (10.8,0) {$y_3$};
	\node[draw, circle](11) at (12.0,0) {};
	\node[draw, circle](12) at (13.2,0) {};
	
	\draw[-, >=latex,thick]  (1)--(2) (2)--(3) (3)--(4) (4)--(5) (5)--(6) (6)--(7) (7)--(8) (8)--(9) (9)--(10) (10)--(11) (11)--(12);
	
	\end{scope}
	
	\begin{scope}[shift={(0,-3)}]
	\draw[fill=black!10, draw=none] (11.4,0.4) rectangle (13.8,-0.4);
	\draw (-0.6,0.4) rectangle (3.6,-0.4);
	\draw (3.6,0.4) rectangle (7.2,-0.4);
	\draw (7.2,0.4) rectangle (10.8,-0.4);
	\draw (10.8,0.4) rectangle (13.8,-0.4);
	
	\node at (-1.1,0) {$A_3$};
	\node[draw, circle](1) at (0,0) {};
	\node[draw, circle](2) at (1.2,0) {};
	\node[draw, circle](3) at (2.4,0) {};
	\node[draw, circle](4) at (3.6,0) {$y_1$};
	\node[draw, circle](5) at (4.8,0) {};
	\node[draw, circle](6) at (6,0) {};
	\node[draw, circle](7) at (7.2,0) {$y_2$};
	\node[draw, circle](8) at (8.4,0) {};
	\node[draw, circle](9) at (9.6,0) {};
	\node[draw, circle](10) at (10.8,0) {$y_3$};
	\node[draw, circle](11) at (12.0,0) {};
	\node[draw, circle](12) at (13.2,0) {};
	
	\draw[-, >=latex,thick]  (1)--(2) (2)--(3) (3)--(4) (4)--(5) (5)--(6) (6)--(7) (7)--(8) (8)--(9) (9)--(10) (10)--(11) (11)--(12);
	
	\end{scope}
	
	\begin{scope}[shift={(0,-4.5)}]
	
	\draw[fill=black!10, draw=none] (7.8,0.4) rectangle (11.4,-0.4);
	\draw (-0.6,0.4) rectangle (3.6,-0.4);
	\draw (3.6,0.4) rectangle (7.2,-0.4);
	\draw (7.2,0.4) rectangle (11.4,-0.4);
	\draw (11.4,0.4) rectangle (13.8,-0.4);
	
	\node at (-1.1,0) {$A_4$};
	\node[draw, circle](1) at (0,0) {};
	\node[draw, circle](2) at (1.2,0) {};
	\node[draw, circle](3) at (2.4,0) {};
	\node[draw, circle](4) at (3.6,0) {$y_1$};
	\node[draw, circle](5) at (4.8,0) {};
	\node[draw, circle](6) at (6,0) {};
	\node[draw, circle](7) at (7.2,0) {$y_2$};
	\node[draw, circle](8) at (8.4,0) {};
	\node[draw, circle](9) at (9.6,0) {};
	\node[draw, circle](10) at (10.8,0) {$y_3$};
	\node[draw, circle](11) at (12.0,0) {};
	\node[draw, circle](12) at (13.2,0) {};
	
	\draw[-, >=latex,thick]  (1)--(2) (2)--(3) (3)--(4) (4)--(5) (5)--(6) (6)--(7) (7)--(8) (8)--(9) (9)--(10) (10)--(11) (11)--(12);
	\draw [decorate,decoration={brace,mirror,amplitude=6pt},xshift=-4pt,yshift=0pt](-0.2,-0.5) -- (2.9,-0.5) node [black,midway,yshift=-0.6cm] {$B_1$};
	\draw [decorate,decoration={brace,mirror,amplitude=6pt},xshift=-4pt,yshift=0pt](4.6,-0.5) -- (6.5,-0.5) node [black,midway,yshift=-0.6cm] {$B_2$};
	\draw [decorate,decoration={brace,mirror,amplitude=6pt},xshift=-4pt,yshift=0pt](8.2,-0.5) -- (10.1,-0.5) node [black,midway,yshift=-0.6cm] {$B_3$};
	\draw [decorate,decoration={brace,mirror,amplitude=6pt},xshift=-4pt,yshift=0pt](11.9,-0.5) -- (13.6,-0.5) node [black,midway,yshift=-0.6cm] {$B_4$};
	
	\end{scope}
	
	\begin{scope}[shift={(0,-6.5)}]
	
	\draw (-0.6,0.4) rectangle (3.0,-0.4);
	\draw (3.0,0.4) rectangle (6.6,-0.4);
	\draw (6.6,0.4) rectangle (11.4,-0.4);
	\draw (11.4,0.4) rectangle (13.8,-0.4);
	
	\node at (-1,0) {$A_*$};
	\node[draw, circle](1) at (0,0) {};
	\node[draw, circle](2) at (1.2,0) {};
	\node[draw, circle](3) at (2.4,0) {};
	\node[draw, circle](4) at (3.6,0) {$y_1$};
	\node[draw, circle](5) at (4.8,0) {};
	\node[draw, circle](6) at (6,0) {};
	\node[draw, circle](7) at (7.2,0) {$y_2$};
	\node[draw, circle](8) at (8.4,0) {};
	\node[draw, circle](9) at (9.6,0) {};
	\node[draw, circle](10) at (10.8,0) {$y_3$};
	\node[draw, circle](11) at (12.0,0) {};
	\node[draw, circle](12) at (13.2,0) {};
	
	\draw[-, >=latex,thick]  (1)--(2) (2)--(3) (3)--(4) (4)--(5) (5)--(6) (6)--(7) (7)--(8) (8)--(9) (9)--(10) (10)--(11) (11)--(12);
	
	\end{scope}
	\end{tikzpicture}
	\caption{Example of the Sperner sequence $A_1,\dots,A_4$ for $n=4$, as well as the derived partition $A_*$. Vertical lines indicate the positions $x_i^1, x_i^2, x_i^3$ of the knives, $i = 1,\dots,4$. Shaded in gray, for $i = 1,\dots,4$, is the bundle $\smash{I_i^{\phi(i)}}$ selected by agent $i$ as their favorite bundle in $A_i$.
		\label{fig:sperner}
	}
\end{figure}
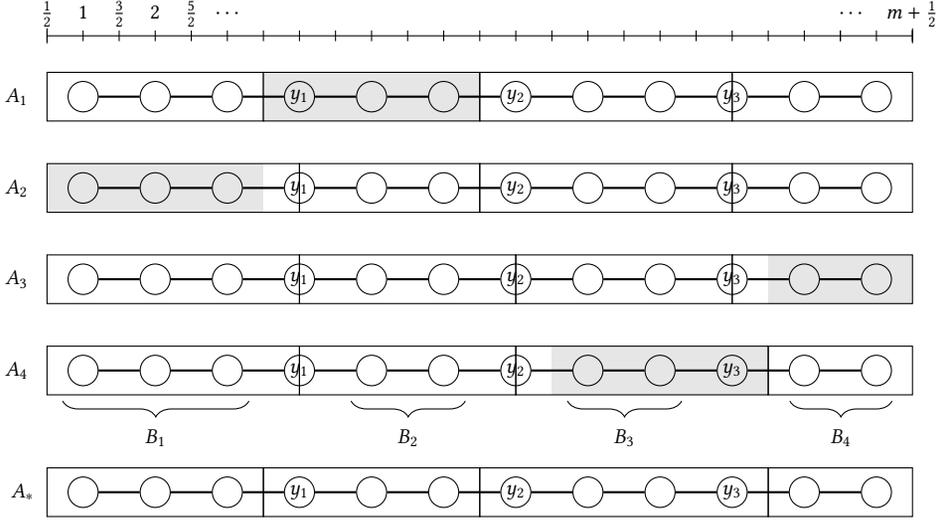

\paragraph{Rounding into a complete partition.}
We now construct a complete partition of the path $P$ into the bundles $(I_*^{1},I_*^{2},\ldots,I_*^{n})$ which are defined as follows:
\[ I_*^j := I_1^j \cup \cdots \cup I_n^j \quad\text{for each $j\in[n]$}.\]
Thus, the bundle $I_*^j$ contains the basic bundle $B^j$, plus all of the boundary items $y^{j-1}$ or $y^j$ that occur in the $j$-th bundle at some point of the Sperner sequence. 
Precisely, for each boundary item $y^j$, $j\in [n-1]$, the item $y^j$ is placed in bundle $I_*^{j+1}$ in case (a) above, and it is placed in bundle $I_*^j$ in case (b).
Thus, every item is allocated to exactly one bundle.

\paragraph{An EF2 allocation.}
We first show that the partition $(I_*^{1},I_*^{2},\ldots,I_*^{n})$ is such that agents' expectations about the value of the bundles $I_*^j$ are approximately correct (up to two items):
\begin{equation}
\label{eq:ef2:approx-correct}
u_i(I_*^j) \ge u_i (I_i^j) \ge u_i(B^j) \quad\text{ for every agent $i\in [n]$ and every $j\in [n]$.}
\end{equation}
This follows by monotonicity of $u_i$, since $I_*^j = I_1^j \cup \cdots \cup I_n^j \supseteq I_i^j \supseteq B^j$ by \eqref{eq:basic-bundle-is-basic}.

Now, based on the partition, we define an allocation $A_*$ by $A_*(i) = I_*^{\phi(i)}$ for each agent $i\in [n]$. Then $A_*$ satisfies EF2:
For any pair $i,j \in [n]$ of agents, we have
\begin{alignat*}{2}
	u_i(A_*(i))
	= u_i(I_*^{\phi(i)})
	&\ge u_i (I_i^{\phi(i)})  &&\text{by \eqref{eq:ef2:approx-correct}}
	\\
	&\ge u_i (I_i^{\phi(j)}) &&\text{by \eqref{eq:ef2:my-bundle-is-best}}
	\\
	&\ge u_i(B^{\phi(j)}) &&\text{by \eqref{eq:ef2:approx-correct}}
	\\
	&= u_i(A_*(j) \setminus \{ y^{j-1}, y^j \}). \qquad\qquad &&\text{by \eqref{eq:basic-bundle-is-basic}}
\end{alignat*}
Hence, we have proved the main result of this section:
\begin{theorem}\label{thm:EF2}
	On a path, for any number of agents with monotone valuation functions, a connected EF2 allocation exists.
\end{theorem}

\section{EF1 existence for four agents}\label{sec:four}
We have seen that Sperner's lemma can be used to show EF2 existence for any number of agents.
Why does our proof in the previous section only establish EF2, and not EF1?
The reason is that agents' expectations about the contents of a bundle might differ by up to \emph{two} goods from what the bundle will actually contain. In the notation of the previous section, an agent $i$ may be presented with a partial partition $I_i$ where the $j$-th bundle $I_i^j$ is the basic bundle, i.e., $I_i^j = B^j$. The agent then selects their favorite bundle from $I_i$, implicitly assuming that the $j$-th bundle in the rounded partition $I_*$ will also equal $B^j$, i.e., that $I_*^j = B^j$. However, it may happen that in fact $I_*^j = \{ y^{j-1} \} \cup B^j \cup \{ y^j \}$, and then $i$ envies the agent who receives bundle $j$ by a margin of two goods.

For four agents, we can adapt our argument to achieve EF1. To do this, we both change the way we round the Sperner sequence into an allocation, and define new labeling functions that better anticipate how a partial partition will be rounded into the final allocation. In this way, agents' expectations about bundles can only be wrong up to one good. In crude terms, agents will expect that each of the two interior bundles will be assigned at least one of the boundary items, and the rounding method ensures that this will indeed happen.

Let $n=4$.
Formally, to define the labeling function, for each agent $i\in [n]$ we construct a \emph{virtual valuation function} $\hat u_i(\bfx, j)$ which assigns a value to each bundle $j\in [n]$ of a partial allocation as specified by a vertex $\bfx \in V(T_\half)$. The way these virtual valuations are defined differs based on the index $j$; in particular, end bundles ($j=1,4$) are treated differently from interior bundles ($j=2,3$). The virtual valuations are defined as follows, for each $\bfx \in V(T_\half)$ and each $i\in [n]$, where the middle row \eqref{eq:virtual23} applies to $j = 2$ and $j = 3$:
\begin{align}
\label{eq:virtual1}
\hat u_i(\bfx, 1) &=
\begin{cases}
u_i(\{1,\dots, x^1 - 1 \}) & \text{if $x^1\in \mathbb Z$,} \\
u_i(\{1,\dots, x^1 - \frac32 \}) & \text{if $x^1\not\in \mathbb Z$.} \\
\end{cases} \\
\label{eq:virtual23}
\hat u_i(\bfx, j) &=
\begin{cases}
u_i^-(\{x^{j-1}, \dots, x^j\}) \:\: & \text{if $x^{j-1} \in \mathbb Z$ and $x^j \in \mathbb Z$,} \\
u_i(I^j(\bfx)) & \text{otherwise}. \\
\end{cases} \\
\label{eq:virtual4}
\hat u_i(\bfx, 4) &=
\begin{cases}
u_i(\{x^3 + 1, \dots, m \}) & \text{if $x^3\in \mathbb Z$,} \\
u_i(\{x^3 + \frac32, \dots, m \}) & \text{if $x^3\not\in \mathbb Z$.} \\
\end{cases}
\end{align}
Thus, for an interior bundle $j=2,3$, if both the items $x^{j-1}$ and $x^j$ to either side of the bundle are covered by a knife, an agent expects that one of these items (the less-valuable one) will be put into bundle $I_*^j$ of the final rounded allocation (recall the definition of $u_i^-$ in equation \eqref{eq:up-to-one-valuation}). For exterior bundles, $j=1$ (resp.\ $j=4$), if the item $x^1$ (resp.\ $x^3$) is not covered by a knife, the agent does not expect the interior item (next to the knife) to belong to the final bundle $I_*^j$, even though it belongs to the observed bundle $I_i^j$. Otherwise, the virtual allocations are equal to $u_i(I^j(\bfx))$, so the agent expects that $I_*^j = I_i^j$. Later, we show that these expectations are correct up to one item.

Using these virtual valuations, we define labeling functions $\hat L_i : V(T_\half) \to [n]$ so that
\[ \hat L_i(\bfx) \in \{ j\in [n] : {\hat u}_i(\bfx, j) \ge {\hat u}_i(\bfx, k) \text{ for all $k\in [n]$} \}. \]
One can check that these valuation functions are still proper.

Again, by Sperner's lemma, there exists an elementary simplex $S^* = \conv(\bfx_1,\bfx_2,\ldots,\bfx_n)$ of the triangulation $T_{\half}$ which is fully-labeled according to our new labeling function: there is a permutation $\phi:[n] \rightarrow [n]$, with $\hat L_{i}(\bfx_{i})=\phi(i)$ for all $i\in [n]$. Again, this elementary simplex induces a Sperner sequence $(A_1, \dots, A_n)$ of partial partitions.

To shorten a case distinction, we assume that $y^2 \in I^2_1 \cup  I^2_2 \cup  I^2_3 \cup I^2_4$, i.e., that the boundary item $y^2$ appears in the second but not in the third bundle in the Sperner sequence. This assumption is without loss of generality, since by the left-right symmetry of the definition of virtual valuations, if necessary we can reverse the path $P$ and consider the same elementary simplex with vertices ordered in reverse ($\bfx_4,\bfx_3, \bfx_2,\bfx_1$); it will still be fully-labeled.

With this assumption made throughout the rest of the argument, we now round the Sperner sequence into a complete partition $(I_*^1,I_*^2,I_*^3,I_*^4)$ of $P$ defined as follows:
\begin{align*}
I_*^1 := I^1_1 \cup \cdots \cup I^1_4, \qquad
I_*^2 := I^2_1 \cup \cdots \cup I^2_4, \qquad
I_*^3 := B^3 \cup \{ y^3 \}, \qquad
I_*^4 := B^4.
\end{align*}
Depending on the placement of the boundary item $y^1$, we will either have $I_*^1 = B^1$ or $I_*^1 = B^1 \cup \{ y^1 \}$; and either $I_*^2 = \{ y^1 \} \cup B^2 \cup \{ y^2 \}$ or $I_*^2 = B^2 \cup \{y^2\}$. With these choices, each interior bundle ($j=2,3$) receives at least one of the boundary items adjacent to it.

The main part of showing that the partition $(I_*^1,I_*^2,I_*^3,I_*^4)$ can be made into an EF1 allocation is an analogue of \eqref{eq:ef2:approx-correct}, which shows that agents' expectations about their bundle are approximately correct. The following analogous proposition is proved by case analysis.
\begin{proposition}
	\label{prop:ef1:approx-correct}
	For each $i\in [n]$ and each $j\in [n]$, we have
	$u_i(I_*^j) \ge \hat u_i(\bfx_i, j) \ge u_i^-(I_*^j)$.
\end{proposition}
\begin{proof}
	We consider each bundle $j = 1,2,3,4$ separately.
	\begin{itemize}[leftmargin=23pt]
		\item Suppose $j = 1$.
		\begin{itemize}
			\item Suppose $x_i^1 \in \mathbb Z$. Then $y^1 = x_i^1$ and $B^1 = \{ 1,\dots, x_i^1 - 1 \}$. Thus $u_i(I_*^1) \ge \hat u_i(\bfx_i, 1) = u_i(B^1) \ge u_i^-(I_*^1)$, since $I_*^1$ is either $B^1$ or $B^1 \cup \{y^1\}$.
			\item Suppose $x_i^1 \not\in \mathbb Z$. Then $\hat u_i(\bfx_i, 1) = u_i(\{ 1, \dots, x_i^1 - \frac32 \})$. Now, either
			\begin{itemize}
				\item $y^1 = x_i^1 - \frac12$ so that $y^1 \in I_i^1$, and so $I_*^1 = B^1 \cup \{ y^1 \} = \{ 1, \dots, x_i^1 - \frac12 \}$, or
				\item $y^1 = x_i^1 + \frac12$ so that $y^1 \not\in I_*^1$, and so $I_*^1 = B^1 = \{ 1, \dots, x_i^1 - \frac12 \}$.
			\end{itemize}
			In either case, $I_*^1 = \{ 1, \dots, x_i^1 - \frac32, x_i^1 - \frac12 \}$, so $u_i(I_*^1) \ge \hat u_i(\bfx_i, 1) \ge u_i^-(I_*^1)$.
		\end{itemize}
		\item Suppose $j = 2$, and suppose that $I_*^2 = B^2 \cup \{ y^2 \}$
		\begin{itemize}
			\item Suppose $x_i^1 \in \mathbb Z$ and $x_i^2 \in \mathbb Z$. So $x_i^1 = y^1$ and $x_i^2 = y^2$. Then $\hat u_i(\bfx_i, 2) = u_i^-(\{y^1, \dots, y^2\})$.  Thus $u_i(I_*^2) \ge \hat u_i(\bfx_i, 2) \ge u_i^-(I_*^2)$ since $I_*^2 = B^2 \cup \{y^2\}$.
			\item Otherwise $\hat u_i(\bfx_i, 2) = u_i(I_i^2(\bfx))$. Since $y^1 \not\in I_1^2(\bfx)$ (because $y^1 \in I^1_*$), we have that $I_i^2(\bfx)$ is either $B^2$ or $B^2 \cup \{y^2\}$.  So $u_i(I_*^2) \ge \hat u_i(\bfx_i, 2) = u_i(I_i^2(\bfx)) \ge u_i^-(I_*^2)$ since $I_*^2 = B^2 \cup \{y^2\}$.
		\end{itemize}
		\item Suppose $j = 2$, and suppose that $I_*^2 = \{ y^1 \} \cup B^2 \cup \{ y^2 \}$.
		\begin{itemize}
			\item Suppose $x_i^1 \in \mathbb Z$ and $x_i^2 \in \mathbb Z$. So $x_i^1 = y^1$ and $x_i^2 = y^2$. Then $\hat u_i(\bfx_i, 2) = u_i^-(\{y^1, \dots, y^2\})$.  Thus $u_i(I_*^2) \ge \hat u_i(\bfx_i, 2) = u_i^-(I_*^2)$ since $I_*^2 = \{y^1\} \cup B^2 \cup \{y^2\}$.
			\item Otherwise $\hat u_i(\bfx_i, 2) = u_i(I_i^2(\bfx))$. First note that $I_i^2(\bfx) \neq B^2$: this is because both $y^1$ and $y^2$ appear in the second bundle of the Sperner sequence (by the case and the wlog assumption), so that $x_i^1 \le y^1$ and $y^2 \le x_i^2$. Since at least one of $x_i^1$ or $x_i^2$ is not integral, at least one of $y^1$ or $y^2$ must be in $I_i^2(\bfx)$. Hence $I_i^2(\bfx)$ is either $\{ y^1 \} \cup B^2 \cup \{ y^2 \}$ or $\{ y^1 \} \cup B^2$ or $B^2 \cup \{ y^2 \}$. In each case, $u_i(I_*^2) \ge \hat u_i(\bfx_i, 2) = u_i(I_i^2(\bfx)) \ge u_i^-(I_*^2)$ since $I_*^2 = \{y^1\} \cup B^2 \cup \{y^2\}$.
		\end{itemize}
		\item Suppose $j = 3$.
		\begin{itemize}
			\item Suppose $x_i^2 \in \mathbb Z$ and $x_i^3 \in \mathbb Z$. So $x_i^2 = y^2$ and $x_i^3 = y^3$. Then $\hat u_i(\bfx_i, 3) = u_i^-(\{y^2, \dots, y^3\})$.  Thus $u_i(I_*^3) \ge \hat u_i(\bfx_i, 3) \ge u_i^-(I_*^3)$ since $I_*^3 = B^3 \cup \{y^3\}$.
			\item Otherwise, since $y^2$ does not appear in $I_1^3(\bfx)$ (by our wlog assumption), we have that $I_i^3(\bfx)$ is either $B^3$ or $B^3 \cup \{y^3\}$.  Now $u_i(I_*^3) \ge \hat u_i(\bfx_i, 3) = u_i(I_i^3(\bfx)) \ge u_i^-(I_*^3)$ since $I_*^3 = B^3 \cup \{y^3\}$.
		\end{itemize}
		\item Suppose $j = 4$.
		\begin{itemize}
			\item Suppose $x_i^3 \in \mathbb Z$. Then $y^3 = x_i^3$ and $B^4 = \{ x_i^3 + 1, \dots, m \}$. Thus $u_i(I_*^4) \ge \hat u_i(\bfx_i, 4) = u_i(B^4) \ge u_i^-(I_*^4)$, since $I_*^4 = B^4$.
			\item Suppose $x_i^3 \not\in \mathbb Z$. Then $\hat u_i(\bfx_i, 4) = u_i(\{ x_i^3 + \frac32, \dots, m \})$. Now, either
			\begin{itemize}
				\item $y^3 = x_i^3 + \frac12$ so $I_*^4 = B^4 = \{ x_i^3 + \frac32, \dots, m \}$, or
				\item $y^3 = x_i^3 - \frac12$ so $I_*^4 = B^4 = \{ x_i^3 + \frac12, \dots, m \}$.
			\end{itemize}
			In either case, $u_i(I_*^4) \ge \hat u_i(\bfx_i, 4) = u_i(\{ x_i^3 + \frac32, \dots, m \}) \ge u_i^-(I_*^4)$.\qedhere
		\end{itemize}
	\end{itemize}
\end{proof}

Now again, based on the partition, we can define an allocation $A_*$ by $A_*(i) = I_*^{\phi(i)}$ for each agent $i\in [n]$. Thus, each agent $i$ receives the bundle in the complete partition corresponding to $i$'s most-preferred index $\phi(i)$.
We prove that $A_*$ satisfies EF1:
For any pair $i,j \in [n]$ of agents, we have
\begin{alignat*}{2}
u_i(A_*(i))
= u_i(I_*^{\phi(i)})
&\ge \hat u_i (\bfx_i, \phi(i)) &&\text{by Proposition~\ref{prop:ef1:approx-correct}}
\\
&\ge \hat u_i (\bfx_i, \phi(j)) &&\text{since $\smash{\hat L_{i}}(\bfx_{i})=\phi(i)$}
\\
&\ge u_i^-(\smash{I_*^{\phi(j)}}) = u_i^-(A_*(j)). \qquad \quad &&\text{by Proposition~\ref{prop:ef1:approx-correct}}
\end{alignat*}
Hence, we have proved the main result of this section:
\begin{theorem}\label{thm:EF1-4agents}
	On a path, for four agents with monotone valuation functions, a connected EF1 allocation exists.
\end{theorem}

For five or more agents, we were not able to construct labeling functions and a rounding scheme which ensure that agents' expectations are correct up to one item. In the four-agent case, each interior bundle is adjacent to an exterior bundle (which helps in the construction), but for five agents, there is a middle bundle whose neighboring bundles are also interior.

\section{EF1 existence for identical valuations}\label{sec:identical}
A special case of the fair division problem is the case of \emph{identical valuations}, where all agents have the same valuation for the goods: for all agents $i,j\in N$ and every bundle $I \in \calC(V)$, we have $u_i(I)=u_j(I)$. We then write $u(I)$ for the common valuation of bundle $I$. The case of identical valuations often allows for more positive results and an easier analysis. Indeed, we can prove that, for identical valuations and \textit{any} number of agents, an EF1 allocation connected on a path is guaranteed to exist and can be found in polynomial time.

Now, one might guess that in the restricted case of identical valuations, egalitarian allocations are EF1. However, the leximin-optimal connected allocation may fail EF1: Consider a path with five items and additive valuations 1--3--1--1--1 shared by three agents. The unique leximin allocation is (1,~3,~1--1--1), which induces envy even up to one good. The same allocation also uniquely maximizes Nash welfare, so the Nash optimum also does not guarantee EF1. In contrast, when requiring bundles to satisfy matroid constraints (rather than connectivity constraints), the Nash optimum is EF1 with identical valuations \citep{Biswas2018}.

Maximizing an egalitarian objective seemed promising because it ensures that no-one is too badly off, and therefore has not much reason to envy others. The problem is that some bundles might be too desirable. To fix this, we could try to reallocate items so that no bundle is too valuable. This is exactly the strategy of our algorithm: It starts with a leximin allocation, and then moves items from high-value bundles to lower-value bundles, until the result is EF1. In more detail, the algorithm identifies one agent $i$ who is worst-off in the leximin allocation, and then adjusts the allocation so that $i$ does not envy any other bundle up to one good. The algorithm does this by going through all bundles in the allocation, outside-in, and if $i$ envies a bundle $I^j$ even up to one good, it moves one item from $I^j$ inwards (in $i$'s direction), see Figure~\ref{fig:leximin-ef1}. As we will show, a key invariant preserved by the algorithm is that the value of $I^i$ never increases, and $i$ remains worst-off. Thus, since $i$ does not envy others up to one good, the allocation at the end is EF1.

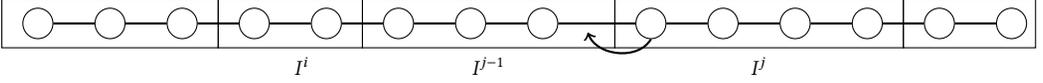
\begin{figure}[ht]
	\centering
	\begin{tikzpicture}[scale=0.8, transform shape, every node/.style={minimum size=5mm, inner sep=1pt}]	
	\node (i)   at (4.4,-0.7) {$I^i$};
	\node (j-1) at (7.5,-0.7) {$I^{j-1}$};
	\node (j)   at (12.0,-0.7) {$I^j$};
	
	\draw (-0.6,0.4) rectangle (3.0,-0.4);
	\draw (3.0,0.4) rectangle (5.4,-0.4);
	\draw (5.4,0.4) rectangle (9.6,-0.4);
	\draw (9.6,0.4) rectangle (14.4,-0.4);
	\draw (14.4,0.4) rectangle (16.6,-0.4);
	
	\node[draw, circle](1) at (0,0) {};
	\node[draw, circle](2) at (1.2,0) {};
	\node[draw, circle](3) at (2.4,0) {};
	\node[draw, circle](4) at (3.6,0) {};
	\node[draw, circle](5) at (4.8,0) {};
	\node[draw, circle](6) at (6,0) {};
	\node[draw, circle](7) at (7.2,0) {};
	\node[draw, circle](8) at (8.4,0) {};
	\begin{scope}[shift={(0.6,0)}]
	\node[draw, circle](9) at (9.6,0) {};
	\node[draw, circle](10) at (10.8,0) {};
	\node[draw, circle](11) at (12.0,0) {};
	\node[draw, circle](12) at (13.2,0) {};
	\node[draw, circle](13) at (14.4,0) {};
	\node[draw, circle](14) at (15.6,0) {};
	\end{scope}

	\draw [->, thick] (9.south) to [out=-120,in=-70] (9.15,-0.15);
	
	\draw[-, >=latex,thick]  (1)--(2) (2)--(3) (3)--(4) (4)--(5) (5)--(6) (6)--(7) (7)--(8) (8)--(9) (9)--(10) (10)--(11) (11)--(12) (12)--(13) (13)--(14);
	\end{tikzpicture}
	\vspace{-8pt}
	\caption{If $i$ envies $j$ even up to one good, Algorithm~\ref{alg:leximin-ef1} takes an item out of bundle $I^j$ and moves it in $i$'s direction.
		\label{fig:leximin-ef1}
	}
\end{figure}

Formally, a \emph{leximin allocation} is an allocation which maximizes the lowest utility of an agent; subject to that it maximizes the second-lowest utility, and so on. In particular, if the highest achievable minimum utility is $u_L$, then the leximin allocation is such that every agent has utility at least $u_L$, and the number of agents with utility exactly $u_L$ is minimum.

\begin{algorithm}[ht]
	\caption{Find a connected EF1 allocation of a path $P$ with identical and monotonic valuations}
	\label{alg:leximin-ef1}
	\begin{algorithmic}[1]
		\REQUIRE a path $P=(v_1,v_2,\ldots, v_m)$, and $n$ agents with identical and monotonic valuations $u$
		\ENSURE an EF1 connected allocation of $P$
		\STATE Let $A = (I^1,\dots,I^n)$ be a leximin allocation of $P$
		\STATE Fix an agent $i$ with minimum utility in $A$, i.e., $u(I^i) \le u(I^j)$ for all $j\in [n]$
		\FOR{$j = 1,\dots,i-1$}
		\IF{$i$ envies $I^j$ even up to one good, i.e., $u(I^i) < u^-(I^j)$}
		\STATE repeatedly delete the rightmost item of $I^j$ and add it to $I^{j+1}$ until $u(I^i) \ge u^-(I^j)$
		\ENDIF
		\ENDFOR
		\FOR{$j = n,\dots,i+1$}
		\IF{$i$ envies $I^j$ even up to one good, i.e., $u(I^i) < u^-(I^j)$}
		\STATE repeatedly delete the leftmost item of $I^j$ and add it to $I^{j-1}$ until $u(I^i) \ge u^-(I^j)$
		\ENDIF
		\ENDFOR
		\RETURN $A$
	\end{algorithmic}
\end{algorithm}

\begin{theorem}
	\label{thm:ef1-identical}
	For identical valuations on a path, Algorithm~\ref{alg:leximin-ef1} finds an EF1 allocation.
\end{theorem}
\begin{proof}
	For an allocation $A = (I^1,\dots,I^n)$, write $u_L(A) := \min_{j\in N} u(I^j)$ for the minimum utility obtained in $A$, and write $L(A) := \{ j \in [n] : u(I^j) = u_L(A) \}$ for the set of agents (\textit{losers}) who obtain this utility. For the leximin allocation $A_{\text{leximin}}$ obtained at the start of the algorithm, write $u_L^* := u_L(A_{\text{leximin}})$ and $L^* := L(A_{\text{leximin}})$. Note that by leximin-optimality, for every allocation $A$ we must have $u_L(A) \le u_L^*$, and if $u_L(A) = u_L^*$ then $|L(A)| \ge |L^*|$. Let $i\in L^*$ be the agent fixed at the start of the algorithm.
	
	\textit{Claim 1.} Throughout the algorithm, $u_L(A) = u_L^*$ and $L(A) = L^*$.
	
	The claim is true before we start the for-loops. Suppose the claim holds up until some iteration of the first for-loop, and we now move an item from $I^j$ to $I^{j+1}$, obtaining the new bundles $I_\text{new}^j$ and $I_\text{new}^{j+1}$ in the new allocation $A_\text{new}$. Then $u(I_\text{new}^j) \ge u^-(I^j) > u(I^i) = u_L^*$, where the strict inequality holds by the if- and until-clauses. Since no agent other than $j$ has become worse-off in $A_\text{new}$, it follows that $u_L(A_\text{new}) \ge u_L(A) = u_L^*$. As noted, by optimality of $u_L^*$, we have $u_L(A_\text{new}) \le u_L^*$. Hence $u_L(A_\text{new}) = u_L^*$. Thus, by optimality of $L^*$, we have $|L(A_\text{new})| \ge |L^*|$. Because agent $j$ has not become a loser (since $\smash{u(I_\text{new}^j)} > u_L^*$ as shown before) and no other agent has become a loser, we have $L(A_\text{new}) \subseteq L(A) = L^*$. Thus $L(A_\text{new}) = L^*$, as required. The second for-loop is handled similarly.
	
	\textit{Claim 2.} After both for-loops terminate, agent $i$ does not envy any agent up to one good.
	
	For any $j\neq i$, agent $i$ does not envy $j$ up to one good immediately after the relevant loop has handled $j$, and at no later stage of the algorithm does $I^j$ change.
	
	It follows that the allocation $A$ returned by the algorithm is EF1: By Claim 1, we have $i\in L(A)$, so that $u(I^j) \ge u(I^i)$ for all $j\in [n]$. By Claim $2$, agent $i$ does not envy any other agent up to one good, so that $u(I^i) \ge u^-(I^k)$ for all $k\in [n]$. Hence, for all $j,k\in [n]$, we have $u(I^j) \ge u^-(I^k)$, that is, no agent envies another agent up to one good.
\end{proof}

Algorithm~\ref{alg:leximin-ef1} can be implemented to run in polynomial time, because with identical valuations, one can use dynamic programming to find a leximin allocation in time $O(m^2n^2)$, and the remainder of Algorithm~\ref{alg:leximin-ef1} takes time $O(mn)$, as each item is moved at most $n$ times. A slight speed-up can be achieved by observing that the proof of Theorem~\ref{thm:ef1-identical} only needed that the initial allocation optimizes the egalitarian welfare $u_L$ and minimizes the cardinality of the set $L$ of losers. Such an allocation can be found by dynamic programming in time $O(m^2n)$, and, after some refinements on the implementation of the dynamic programming approach, the running time can be lowered to $O(mn)$ (see Algorithm \ref{alg:stronglyMMS} in the appendix). 

The reallocation stage of our algorithm bears some similarity to \citeauthor{Suksompong2017}'s \citeyearpar[Thm.~2]{Suksompong2017} proof that a $u_\text{max}$-equitable allocation exists.
\citet[Lem.~C.2]{Oh2018} proved independently, using an inductive argument, that EF1 allocations on a path exist for identical valuations, and can be found in polynomial time.
More recently, \citet{misra2021equitable} presented another algorithm for this task in the context of aiming for connected allocations satisfying \emph{equitability up to one good} (EQ1).

\section{Concluding remarks}
We have studied the existence of EF1 allocations under connectivity constraints imposed by an undirected graph. We have shown that for two, three, or four agents, an EF1 allocation exists if the graph is traceable and if the agents have monotone valuations. For any number of agents, we also proved that traceable graphs guarantee the existence of an EF2 allocation. The latter two results are proved using Sperner's lemma, which has been used many times in economics and game theory to show the existence of equilibria and fair allocations \citep{Scarf1982,Su1999}. Unusually, in our application we were able to use Sperner's lemma in a setting with indivisibilities. 
We leave as an open question whether EF1 allocations on a path exist for five and more agents.

Our procedures for identical valuations as well as for the cases of two or three agents can be efficiently implemented, so that we can find EF1 allocations in polynomial time.
For our results based on Sperner's lemma, it is not clear how to compute EF1 and EF2 allocations efficiently.
On the other hand, as is the case with the computation of other structures whose existence follows from Sperner's lemma (such as Nash equilibrium or more broadly PPAD problems), our proof based on Sperner's lemma allows for a “path-following” algorithm through the subdivided simplex. This type of algorithm has been observed to be practically efficient in other contexts \citep{scarf1967approximation}, and may also be efficient in our allocation setting on practical instances. Formally, we do not know of a computational complexity result for the problem of finding EF1 or EF2 allocations on a path. For divisible cake-cutting, it is PPAD-complete to find an $\varepsilon$-approximate envy-free allocation \citep{Deng2012}, implying that it is unlikely that there is an algorithm that runs in time polynomial in $n$ and $\log \frac{1}{\varepsilon}$.
However, the PPAD-hardness proof of \citet{Deng2012} uses non-monotone valuations, and thus does not easily extend to our setting, where we assume monotone valuations.
For general graphs, \citet[Thm.~2]{parameterizedFair} recently showed that deciding whether an EF1 allocation exists on a given instance is NP-hard. Their result applies when the underlying graph is a star, and holds even if the $n$ agents have binary additive valuations.

Regarding strategic aspects, existing results from the literature imply that there are no EF1 allocation rules which are strategyproof.%
\footnote{In this paragraph, we follow the exposition of \citet[Chapter~12]{peters2019fair}.}
\citet{amanatidis2017truthful} characterized all strategyproof allocation mechanisms when there are $n = 2$ agents with additive valuations over indivisible items (with no connectivity constraints). They then proved that no mechanism in their class guarantees EF1 \citep[Sec.~4.2]{amanatidis2017truthful} for $m \ge 5$ items. It follows that there is also no strategyproof EF1 mechanism that respects connectivity constraints. We can also obtain such a result by reduction from divisible cake-cutting. Fix some $\epsilon > 0$, and suppose we had a mechanism for allocating a path of $M$ items among $n$ agents while being strategyproof and EF1. Then we can use this mechanism as a mechanism for cake-cutting: Given continuous agent valuations over the interval $[0,1]$, approximate these by additive valuations over the path of $M$ items and run the mechanism on this instance. For sufficiently large $M$, the resulting mechanism for cake-cutting will be $\epsilon$-strategyproof (in the sense that a misreport can increase utility by at most $\epsilon$) and $\epsilon$-envy-free (in the sense that envy is bounded by $\epsilon$). However, the literature on cake-cutting contains impossibilities about strategyproofness and envy-freeness when requiring connected pieces (\citealp{bei2017cake}, Theorem 1, \citealp{bei2018truthful}, Theorem 3), and the proofs also establish impossibility for the $\epsilon$-versions of these properties for small enough $\epsilon$. Hence, for large enough $M$, no strategyproof EF1 mechanism for the indivisible setting can exist. By using the result of \citet{bei2018truthful}, we can obtain an impossibility for $n = 2$ and even for binary additive valuations where each agent approves an interval of items beginning with the left-most item. Finally, \citet[Chapter~12]{peters2019fair} gives a simple direct proof that there are no strategyproof EF1 mechanisms, even for $n = 2$ agents and $m = 5$ items on a line.

We gave a forbidden minor type characterization of all graphs that guarantee the existence of EF1 allocations for two agents.
It is natural to also consider the case of more than two agents. However, there are several difficulties in extending the characterization result beyond two agents. 
First, one cannot generalize the notion of a bipolar ordering in a meaningful way; indeed, a \emph{tripolar ordering}, requiring each initial, middle, and last segment to be connected in a given graph, reduces to the notion of a Hamiltonian path because every consecutive pair of such an ordering must be connected. Second, our two-agent characterization heavily depends on having a simple envy-free protocol: the cut-and-choose procedure. Unfortunately, for three agents, the known protocol becomes much more complex (see Section~\ref{sec:three}), and for four agents, there is no known explicit protocol that constructs an envy-free division. Nevertheless, \citet{ZwickerIgarashi2021} recently proposed a forbidden minor type conjecture for the continuous variant of our problem. It would be interesting to explore the discrete version of their conjecture.

In the setting without connectivity constraints, it is possible to achieve efficiency and fairness simultaneously: the maximum Nash welfare solution yields an allocation that is both EF1 and Pareto-optimal \citep{CKM+16a}. In our model, this is unfortunately impossible, since on a path there are instances where there is no connected allocation which is EF1 and Pareto-optimal, and it is NP-hard to decide whether such an allocation exists \citep{Igarashi2018}.

In this paper, we have only considered \emph{goods}, with monotonic valuations. The setting where some or all items are undesirable (so-called \emph{chores}) is also of interest \citep{AzizCI18,BMSY16a,MeSh18a,Halevi2018,bouveret2019chore,hohne2021allocating}. On a path, a connected allocation satisfying proportionality up to one good (PROP1) always exists \citep{AzizCI18}, but the existence of EF1 or EF2 allocations in this domain is open. For cake-cutting, when agents consider some parts of the cake undesirable, Sperner's lemma does not directly produce a connected envy-free allocation \citep{Halevi2018}, but other methods can prove the existence of such allocations in most cases \citep{Halevi2018,MeSh18a}.

\bibliographystyle{plainnat}

\clearpage
\appendix
\section{Appendix}

\subsection{Maximin share allocations}
\label{sec:mms}
The \emph{maximin share guarantee} of an agent $i\in N$
is
\[
\MMS_i : = \max_{(P^1,P^2,\dots, P^n)\in\Pi_n}\min_{j\in [n]} u_i(P^j),
\]
where $\Pi_n$ denotes the space of all partitions of $V$ into $n$ connected bundles. An allocation $A$ is a \emph{maximin share (MMS) allocation} if $u_{i}(A(i)) \ge \MMS_{i}$ for each agent $i \in N$. (Note that the maximum is taken only over connected partitions, so the MMS value could be lower than the standard definition from the model without connectivity constraints.) \citet{Bouveret2017} showed that an MMS allocation exists if the underlying graph $G$ is a tree. On the other hand, MMS allocations need not exist on a cycle \citep{Bouveret2017,Lonc2018} or on a complete graph \citep{Kurokawa2018}. The computational complexity of determining the existence of MMS allocations under several graph constraints has also been investigated \citep{Greco2020}.

Since we have seen that EF1 or EF2 allocations are guaranteed to exist on a path, it is natural to ask whether we can additionally require MMS: on a path, does there always exist an allocation that satisfies EF1 and MMS?

First, let us note that not every EF1 allocation is also MMS. For 3--1--1--1--3 and three agents, the MMS value is 3 via the partition (3,~1--1--1,~3), but the EF1 allocation (3--1,~1,~1--3) gives the middle agent a utility of only 1. In fact, one can show that this example is worst possible, for subadditive valuations. Valuations $u_i$ are \emph{subadditive} if, for any bundles $I,I'$, we have $u_i(I\cup I') \le u_i(I) + u_i(I')$.
An allocation satisfies $\alpha$-MMS for some $\alpha > 0$ if $u_{i}(A(i)) \ge \alpha \cdot \MMS_{i}$ for each agent $i \in N$. As MMS allocations need not exist in general, $\alpha$-MMS allocations have been widely investigated \citep{Amanatidis2017, Amanatidis2018, Kurokawa2018}. 

\begin{proposition}
	\label{prop:ef1-1-3-mms}
	For subadditive valuations, an EF1-allocation on a path guarantees 1/3-MMS.
\end{proposition}
\begin{proof}
	Let $A$ be an EF1 allocation, write $I^j = A(j)$ for all $j\in [n]$, and fix some agent $i$. For each $j \in [n] \setminus \{i\}$, let $g_j\in I^j$ be an item such that $u_i(I^i) \ge u_i(I^j \setminus \{g_j\})$. We show that $u_i(I^i) \ge \frac13 \MMS_i$.
	
	Let $P = (P^1,\dots,P^n)$ be a partition of the items into $n$ bundles such that $u_i(P^j) \ge \MMS_i$ for each $j\in [n]$. Since there are $n$ bundles in $P$ but only $n-1$ items $g_j$, there must be some bundle $P^k$ such that $g_j\not\in P^k$ for all $j \in [n] \setminus \{i\}$; we show that $u_i(P^k) \le 3\cdot u_i(I^i)$.
	
	Suppose for a contradiction that there are three distinct agents $j_1,j_2,j_3 \in [n] \setminus \{i\}$ such that $P^k\cap I^{j_r} \neq \emptyset$ for $r=1,2,3$. Since $P^k$ and the $I^{j_r}$'s are all intervals of a path, the middle interval must be completely contained in $P^k$, that is, $I^{j_r} \subseteq P^k$ for some $r$. Hence $g_{j_r} \in P^k$, contradicting the choice of $P^k$. So $P^k$ intersects at most two bundles from $A$ other than $I^i$. Thus, for some $j_1,j_2 \in [n] \setminus \{i\}$, we have $P^k \subseteq I^{j_1} \cup I^i \cup I^{j_2} \setminus \{g_{j_1}, g_{j_2} \}$, and thus by subadditivity,
	\[ u_i(P^k) \le u_i(I^{j_1} \setminus \{g_{j_1}\}) + u_i(I^i) + u_i(I^{j_2} \setminus \{g_{j_2} \}) \le 3\cdot u_i(I^i).  \]
	
	Hence, we have $u_i(I^i) \ge \frac13 u_i(P^k) \ge \frac13 \MMS_i$, as required.
\end{proof}

Interestingly, using a similar proof, one can show that the two agents receiving the outer bundles of the path both get at least half of their MMS value. This is also tight; consider 1--1--2--2--1--1 for four agents, and the EF1 allocation (1,1--2,2--1,1).

If we do not restrict valuations to be subadditive, then EF1 does not guarantee $\alpha$-MMS for any $\alpha > 0$: Consider a path $P=(v_1,v_2,v_3)$ of three items, and two agents with identical valuations $u$ defined so that $u(I) = 1$ if $I \supseteq \{v_1,v_2\}$ or $I \supseteq \{v_3\}$, and $u(I) = 0$ otherwise. Then the MMS value is 1 via the partition ($v_1$--$v_2$, $v_3$), but the allocation ($v_1$, $v_2$--$v_3$) is EF1 and gives the left agent utility 0.

For graphs that are not paths, Proposition~\ref{prop:ef1-1-3-mms} does not hold. For a complete graph (i.e., in the absence of connectivity constraints), EF1 only implies $1/n$-MMS \citep{CKM+16a,Amanatidis2018}.

While we have seen that EF1 on a path does not immediately imply MMS, it does imply MMS in many cases.
The following lemma will be useful to show that the allocations produced by our arguments in the main text all satisfy the MMS guarantee.

\begin{lemma}
	\label{lem:mms-n-1}
	Suppose there are $n\ge 2$ agents, and the items are arranged on a path. Take any $n-1$ items $y^1 < \dots < y^{n-1}$, and define the bundles $B^1, \dots, B^n$ as follows:\small
	\[ B^1 = L(y^1), \:
	B^2 = P(y^1 + 1, y^2 - 1), \dots,
	B^{n-1} = P(y^{n-2}+1, y^{n-1}-1), \:
	B^n = R(y^{n-1}). \]\normalsize
	Then for any agent $i$, there is some $r \in [n]$ such that $u_i(B^r) \ge \MMS_i$.
\end{lemma}
\begin{proof}
	Let $P = (P^1,\dots,P^n)$ be a connected partition of the items (ordered left-to-right) so that $u_i(P^j) \ge \MMS_i$ for all $j\in [n]$. Since there are $n$ bundles in $P$ but only $n-1$ items $y^1,\dots,y^{n-1}$, there exists a bundle $P^k$ in $P$ that does not contain any $y^j$.
	Writing $Y = \{y^1,\dots,y^{n-1}\}$, we see that there is some $r \in [n]$ such that\small
	\[ (P^1 \cup \dots \cup P^{k-1}) \cap Y = \{y^1,\dots,y^{r-1}\}
		\quad \text{and} \quad
		(P^{k+1} \cup \dots \cup P^n) \cap Y = \{y^r,\dots,y^{n-1}\}.
	\]\normalsize
	Thus, we have $P^k \subseteq P(y^{r-1}+1, y^{r}-1) = B^r$ so that $u_i(B^r) \ge u_i(P^k) \ge \MMS_i$.
\end{proof}

\begin{theorem}
	For a path, the EF1 allocations constructed by any of our methods guarantee MMS.
\end{theorem}
\begin{proof}
\textit{Discrete cut-and-choose protocol for two agents.}
Suppose Alice's lumpy tie is $v_j$. Then, using the definition of lumpy tie, a connected partition witnessing Alice's MMS value is either $P_1 = (L(v_j), R(v_j) \cup \{ v_j \})$ or $P_2 = (L(v_j) \cup \{ v_j \}, R(v_j))$. At the end of the procedure, Alice receives either $L(v_j) \cup \{ v_j \}$ or $R(v_j) \cup \{ v_j \}$. For either of these options, there is a bundle in $P_1$ and a bundle in $P_2$ which are weakly worse. So Alice receives a bundle that satisfies her MMS value. For Bob, he receives his preferred bundle among $L(v_j)$ or $R(v_j)$. These two bundles are of the shape described in Lemma~\ref{lem:mms-n-1} with $y^1 = v_{j}$, so Bob's choice satisfies his MMS value.

\textit{Moving-knife protocol for three agents.}
The allocation returned by the algorithm of Theorem~\ref{thm:EF1-3agents} guarantees MMS. To see this, first suppose the algorithm terminates in Step 3 or Step 4(a). At that step, the bundles $L$, $M$, and $R$ are of the shape described in Lemma~\ref{lem:mms-n-1} with $y^1 = v_{\ell+1}$ and $y^2 = v_r$. The agent $\shoutleft$ who receives $L$ thinks that $L$ is best among $L,M,R$ (since he shouted), so by Lemma~\ref{lem:mms-n-1} agent $\shoutleft$ receives his MMS value. The proof of Theorem~\ref{thm:EF1-3agents} shows that no other agent envies $\shoutleft$ (even without removing an item), so that every other shouter receives value at least $u_i(L)$, meaning that player $\shouter$ receives at least her MMS value (since she shouted, she finds $L$ weakly better than $M$ and $R$). Finally, agent $\chooser$ does not envy any other agent (even without removing an item), and so automatically receives at least his MMS value.
Next, suppose the algorithm terminates in Step 2. As the proof of Theorem~\ref{thm:EF1-3agents} shows, the agent $\shoutleft$ does not envy either of the other players, and hence receives at least his MMS value. Apply Lemma~\ref{lem:mms-n-1} with $y^1 = v_{\ell}$ and $y^2 = v_r$. It follows that for every agent, one of $L \setminus \{ v_\ell \}$, $M$, or $R$ provides at least the MMS value. But we know from the proof that any agent prefers either $M$ or $R$ to $L \setminus \{ v_\ell \}$, and we know that all agents $i \neq \shoutleft$ receive a bundle weakly preferred to both $M$ and $R$, so $i$ receives his MMS value.
Finally, suppose the algorithm terminates in Step 4(b). Similarly to the argument for Step 2, apply Lemma~\ref{lem:mms-n-1} with $y^1 = v_{\ell+1}$ and $y^2 = v_r$. Then, each player thinks that one of $L$, $M$, $R$ provides the MMS value. By the proof of Theorem~\ref{thm:EF1-3agents}, agent $\shoutleft$ receives a bundle weakly preferred to each of $L,M,R$, and agents $i \neq \shoutleft$ prefer either $M$ or $R$ to $L$, and receive a bundle that is weakly preferred to both $M$ and $R$, so they also receive their MMS value.

\textit{Identical valuations.} Algorithm~\ref{alg:leximin-ef1} gives each agent a utility of at least $u_L^*$. By their definitions, the MMS-value is the same as the optimal egalitarian welfare under identical valuations.

\emph{EF2 via Sperner's lemma.} For each agent $i$, by Lemma~\ref{lem:mms-n-1}, there exists a basic bundle whose value is at least $\MMS_i$. We showed that the allocation $A_*$ is such that agent $i$ weakly prefers the bundle $i$ receives in $A_*$ to any basic bundle. Hence, $A_*$ is an MMS allocation.

\emph{EF1 for four agents via Sperner's lemma.} For each vertex $\bfx_i$ of the full-labeled simplex $S^*$, invoke Lemma~\ref{lem:mms-n-1} with $y^1 = \lfloor x_i^1 \rfloor$, $y^2 = \lfloor x_i^2 \rfloor$, $y^3 = \lceil x_i^3 \rceil$. By case-analysis one can check that $\hat u_i(\bfx_i, j) \ge u_i(B^j)$ for each $j = 1,2,3,4$, where the $B^j$'s are defined like in Lemma~\ref{lem:mms-n-1}. By Proposition~\ref{prop:ef1:approx-correct}, we have that $u_i(A_*(i)) \ge \hat u_i(\bfx_i, \phi(i)) = \max_{j\in [n]} \hat u_i(\bfx_i, j) = \max_{j\in [n]} u_i(B^j) \ge \MMS_i$.
\end{proof}

\subsection{Example of an instance with no EFX allocation}
We define EFX as follows.
\begin{definition}[EFX]
An allocation $A$ satisfies \emph{EFX (Envy-freeness up to any outer good)} if the envy is bounded up to the least valuable outer good, i.e., for any pair $i,j \in N$ of agents, and for every good $u \in A(j)$ such that $A(j) \setminus \{u\}$ is connected, we have $u_i(A(i)) \ge u_i(A(j)\setminus \{u\})$.
\end{definition}

\begin{example}\label{ex:EFX}
Consider the instance 2--3--1--3 for three agents. This instance admits no connected EFX allocation: It is clear that no allocation in which some bundle is empty satisfies EFX. In (2, 3, 1--3), the left agent envies the right agent even after removing the outer good of value 1; in (2, 3--1,3), the left agent envies the middle agent even after removing the outer good of value 1; and in (2--3, 1,3), the middle agent envies the left agent even after removing the outer good of value 3. One can also consider the instance 1--1--3--3 for two agents. \qed
\end{example}

\subsection{Example of a non-traceable graph that guarantees EF1}
\begin{example}\label{ex:EF1:nonHamiltonian}
Consider a star with three leaves. We will divide the graph among three agents. Consider the allocation where each agent chooses the most favorite leaf-vertex among the unallocated vertices in order, with the last agent in that order being assigned to the central vertex of the star. The resulting allocation satisfies EF1, since the envy towards agents allocated to a single item can be bounded up to one good, and the first and second agent do not envy the third agent if one removes the central vertex from his bundle.
\qed
\end{example}

\subsection{Efficient computation of SMMS allocations for identical valuations}\label{app:dyn}
In this section, we discuss how to compute the initial allocation that is needed for Algorithm~\ref{alg:leximin-ef1} to obtain an EF1 allocation under identical valuations (Theorem~\ref{thm:ef1-identical}).

Given a path $P=(v_1,v_2,\ldots,v_m)$, $n$ agents with identical  monotonic valuations $u$, and an allocation $A = (I^1,\dots,I^n)$, write $u_L(A) := \min_{j\in N} u(I^j)$ to denote the \emph{egalitarian welfare} of $A$, i.e., the minimum valuation in $A$ among all agents, and $L(A) := \{ j \in [n] : u(I^j) = u_L(A) \}$ for the set of agents who obtain this minimum utility. We refer to the agents in $L(A)$ as the \textit{losers}.
\begin{definition}[SMMS allocation]
	An allocation of a path for agents with identical monotonic valuations satisfies \emph{strong maximin share (SMMS)} if it minimizes the number of losers among all allocations maximizing the egalitarian welfare. 
\end{definition}
In the following, we provide an efficient dynamic programming algorithm for computing an SMMS allocation. We first present a simple algorithm that runs in time $O(m^2n)$, and then we refine it to improve the running time to $O(mn)$.

\subsubsection*{A suboptimal algorithm.}

In order to use a dynamic programming approach, we start by deriving a recurrence relation characterizing the SMMS allocations. 

Given two partial allocations $A,A'$, we write $A\succeq A'$ if either $u_L(A)>u_L(A')$ holds, or both $u_L(A)=u_L(A')$ and $|L(A)|\leq |L(A')|$ hold; furthermore, we write $A\sim A'$ if both $A\succeq A'$ and $A\preceq A'$ hold. We observe that an allocation $A$ of path $P$ for $n$ agents is SMMS iff it is ``optimal'' according to the ordering relation $\succeq$, i.e., iff $A\succeq A'$ for any allocation $A'$ (of path $P$ for $n$ agents).  

For any $h\in [m+1]$ and $j\in [m]$, let $u[h,j]:=u(P(v_h,v_j))$ be the utility assigned by $u$ to the path segment from $v_h$ to $v_j$. If $h > j$ then we use the convention that $P(v_h,v_j)=\emptyset$ and $u[h,j]=0$. Given $i\in [n]$ and $j\in [m]\cup\{0\}$, let  $A[i,j]$ be an SMMS allocation of the subpath $P(v_1,v_j)$ for $i$ agents. Also, let $\Egal[i,j]$ and $L[i,j]$ denote the egalitarian welfare and the number of losers of the SMMS allocation $A[i,j]$. Note that our ultimate aim is to find $A[n,m]$.

For any $i\in [n]\setminus\{1\}$, $j\in [m]\cup \{0\}$ and $h\in [j+1]$, let $A[i,h,j]$ be an allocation of the subpath $P(v_1,v_j)$ for $i$ agents that is optimal according to $\succeq$ after constraining the $i$-th bundle to be equal to the subpath $P(v_h,v_j)$; furthermore, let $\Egal[i,h,j]$ and $L[i,h,j]$ denote the egalitarian welfare and the number of losers of allocation $A[i,h,j]$. Observe that we allow $h$ to reach the value $j+1$ in order to model the case in which the $i$-th agent gets an empty bundle. By definition, it holds that
\begin{align}
	\Egal[i,h,j]&=\min\{\Egal[i-1,h-1],u[h,j]\},\label{recur_1M}\\
	L[i,h,j]&=
	\begin{cases}L[i-1,h-1]&\text{ if }\Egal[i-1,h-1]<u[h,j],\\
		L[i-1,h-1]+1&\text{ if }\Egal[i-1,h-1]=u[h,j],\\
		1&\text{ if }\Egal[i-1,h-1]>u[h,j],
	\end{cases}\label{recur_1L}
\end{align}
where $\Egal[1,h-1]=u[1,h-1]$ and $L[1,h-1]=1$.

One can easily observe that, for any fixed $i\in [n]$ and $j\in [m]\cup\{0\}$, an optimal allocation $A[i,j]$ of subpath $P(v_1,v_j)$ for $i$ agents can be computed according to the following recurrence relation:
\begin{equation}\label{recur_1A}
	A[i,j]:=\begin{cases}
		\emptyset &\text{if }j=0,\\
		P(v_1,v_j) &\text{if }i=1,\\
		\text{the best allocation in } \{A[i,h,j]:h\in [j+1]\}&\text{otherwise},
	\end{cases}
\end{equation}
where the quality of each allocation $A[i,h,j]$ depends on $\Egal[i,h,j]$ and $L[i,h,j]$ only. 

The recurrence relation \eqref{recur_1A} can be used to design a dynamic programming algorithm that computes the SMMS allocation $A[n,m]$ in time $O(m^2n)$. To do this, we will iteratively compute an integer $k[i,j]$ such that $A[i,j]\sim A[i,k[i,j],j]$ for all $i\in [n]$ and for all $j\in [m]$. 

To do this, in each round $(i,j)$, we identify the allocations $A[i,h,j]$ for each $h\in [j+1]$, and compute the corresponding values $\Egal[i,h,j]$ and $L[i,h,j]$ using \eqref{recur_1M} and \eqref{recur_1L}. Using the computed values $\Egal[i,h,j]$ and $L[i,h,j]$ we can then use \eqref{recur_1A} to find the index $k[i,j]$ such that $A[i,k[i,j],j] = A[i,j]$ is optimal, and we store the resulting values $\Egal[i,j]$ and $L[i,j]$ (which will be used in the subsequent rounds). Then we proceed to the next round. Finally, at the end of the last round $(n,m)$, we can recursively reconstruct the optimal allocation $A[n,m]$ by using the indices of type $k[i,j]$ previously stored.

By \eqref{recur_1A}, the resulting algorithm returns an SMMS allocation, and its time complexity is $O(m^2n)$, given by  the number of rounds (which is the number of pairs $(i,j)$ which is $O(nm)$) multiplied by the complexity of each round (checking each value of $h$ which is in $O(m)$).

\subsubsection*{An improved algorithm.}
By exploiting the monotonicity properties of the valuation functions, we can improve the above algorithm and  lower its running time. In particular, we will see how to execute each round $(i,j)$ in constant amortized time, thus lowering the overall time complexity to $O(mn)$. 

The problem with the existing algorithm is we need to check all possible values of $k[i,j]$. Instead we will introduce three quantities $k_1[i,j]$, $k_2[i,j]$, and $k_3[i,j]$, each of which can be computed quickly, and prove that one of the three values provides a suitable value of $k[i,j]$.  Specifically, for any $i\in [n]$ and $j\in [m]$, let\footnote{In the definition of $k_1[i,j]$, if the set $D:=\{h\geq 1:\Egal[i-1,h-1]<u[h,j]\}$ is empty, taking the maximum between $\sup D$ and $1$ guarantees that $k_1[i,j]$ is equal to $1$ in this extreme case.}
\begin{align*}
	&k_1[i,j]:=\max\{\sup\{h\geq 1:\Egal[i-1,h-1]<u[h,j]\},1\},\nonumber\\
	&k_2[i,j]:=\max\{h\geq 1:\Egal[i-1,h-1]\leq u[h,j]\},\nonumber\\
	&k_3[i,j]:=k_2[i,j]+1.
\end{align*}
By the monotonicity of the utility function, we have that, for any fixed value of $h$, $u[h,j]$ is non-decreasing in $j$. This implies that for $t = 1,2,3$ and fixed $i \in [n]$, the value $k_t[i,j]$ is non-decreasing in $j\in [m]$. Hence the integers of type $k_t[i,j]$ can be recursively written as
\begin{align}
	&k_1[i,j]=\max\{\sup\{h\geq k_1[i,j-1]:\Egal[i-1,h-1]<u[h,j]\},k_1[i,j-1]\},\nonumber\\
	&k_2[i,j]=\max\{h\geq k_2[i,j-1]:\Egal[i-1,h-1]\leq u[h,j]\},\nonumber\\
	&k_3[i,j]=k_2[i,j]+1,\label{k1}
\end{align}
where $k_t[i,0]:=1$ for $t = 1,2,3$. 

In Lemma \ref{lem:alg2} we will show that $A[i,j]$ can be set equal to the best allocation among the three allocations of type $A[i,k_t[i,j],j]$ (with $t=1,2,3$). We first outline some preliminary properties in Lemma \ref{lem:alg1}.
\begin{lemma}\label{lem:alg1}
	For any $i\in [n]$ and $j\in [m]$, we have
	\begin{enumerate}
		\item[(i)] $A[i,j]\succeq A[i,j-1]$
		\item[(ii)] $\Egal[i,h,j]$ is non-decreasing in $h\leq k_1[i,j]$, it is constant in $k_1[i,j]<h\leq k_2[i,j]$, and it is non-increasing in $h>k_2[i,j]$.
	\end{enumerate}
\end{lemma}
\begin{proof}
	We first show (i). Given $i\in [n]$ and $j\in [m]$, let $A'[i,j]$ be the allocation obtained from $A[i,j-1]$ by adding item $j$ to the last bundle of $A[i,j-1]$. By the optimality of $A[i,j]$, we have that $A[i,j]\succeq A'[i,j]\succeq A[i,j-1]$, and this shows (i). 
	
	Now, we show (ii). We have the following properties: (a) $\Egal[i-1,h-1]$ is non-decreasing in $h$ (by (i)), and (b) $u[h,j]$ is non-increasing in $h$ (by the monotonicity of the valuation function). Thus, we get the following additional properties, that immediately imply (ii):
	\begin{itemize}
		\item $\Egal[i-1,h-1]<u[h,j]$ for any $h\leq k_1[i,j]$ (by definition of $k_1[i,j]$ and because of (a) and (b)), and then $\Egal[i,h,j]=\Egal[i-1,h-1]$ is non-decreasing in $h\leq k_1[i,j]$ (by (a));
		\item $u[h,j]<\Egal[i-1,h-1]$ for any $h> k_2[i,j]$ (by definition of $k_2[i,j]$ and because of (a) and (b)), and then $\Egal[i,h,j]=u[h,j]$ is non-increasing in $h> k_2[i,j]$ (by (b));
		\item $\Egal[i,h,j]=\Egal[i-1,h-1]=u[h,j]$ in $k_1[i,j]<h\leq k_2[i,j]$ (by definition of both $k_1[i,j]$ and $k_2[i,j]$), and then $\Egal[i,h,j]$ is necessarily constant in $k_1[i,j]<h\leq k_2[i,j]$ (by (a) and (b)). \qedhere
	\end{itemize}
\end{proof}
\begin{lemma}[Characterization of SMMS Allocations]\label{lem:alg2}
	Given $i\in [n]$ and $j\in [m]$, at least one index $k[i,j]\in \{k_1[i,j],k_2[i,j],k_3[i,j]\}$ guarantees that $A[i,k[i,j],j]$ is an SMMS allocation of subpath $P(v_1,v_j)$ for $i$ agents (i.e., $A[i,j]\sim A[i,k[i,j],j]$). 
\end{lemma}
\begin{proof}
	Let $i\in [n]$, $j\in [m]$, and let $k\in [j+1]$ be an index such that $A[i,k,j]$ is optimal (i.e., SMMS).  We consider  three cases, depending on the value of $k$.
	\begin{enumerate}
		\item[(a)] Suppose $k\leq k_1[i,j]$. By exploiting the monotonicity properties of Lemma \ref{lem:alg1} and the definition of $k_1[i,j]$, we will show that $A[i,k_1[i,j],j]\succeq A[i,k,j]$. By Lemma \ref{lem:alg1}, the set $H_1$ of integers $h\leq k_1[i,j]$ such that $\Egal[i,h,j]=\Egal[i,j]$ is an integer interval having $k_1[i,j]$ as its maximum, thus both $k$ and $k_1[i,j]$ belong to $H_1$.
		Furthermore, for any $h\in H_1$, the egalitarian welfare of each allocation $A[i,h,j]$ is equal to that of allocation $A[i-1,h-1]$, and the set of losers in $A[i,h,j]$ is the same as in $A[i-1,h-1]$ (indeed, agent $i$ is not a loser since   $\Egal[i-1,h-1]<u[h,j]$). Thus, since $A[i-1,k_1[i,j]-1]\succeq A[i-1,h-1]$ (by Lemma \ref{lem:alg1}(i)), we necessarily have that $A[i,k_1[i,j],j]\succeq A[i,h,j]$ for any $h\in H_1$, and this shows the optimality of allocation $A[i,k_1[i,j],j]$. 
		\item[(b)] Suppose $k_1[i,j]<k\leq k_2[i,j]$. Then we can analogously show that $A[i,k_2[i,j],j]\succeq A[i,k,j]$. The set $H_2$ of values $h$ with $k_1[i,j]<h\leq k_2[i,j]$ and $\Egal[i,h,j]=\Egal[i,j]$ is an integer interval having $k_2[i,j]$ as its maximum, thus both $k$ and $k_2[i,j]$ belong to $H_2$. For any $h\in H_2$, the valuation $u[h,j]$ for the $i$-th bundle of $A[i,h,j]$ is equal to the egalitarian welfare  $\Egal[i,h-1]$ of the allocation restricted to the first $i-1$ bundles (by definition of $k_1[i,j]$ and $k_2[i,j]$, and by Lemma~\ref{lem:alg1}(ii)). Thus, the number of losers $L[i,h,j]$ in allocation $A[i,h,j]$ is equal to $L[i-1,h-1]$ (i.e., the quantity of losers among the first $i-1$ agents) plus $1$ (i.e., agent $i$). As $A[i-1,k_2[i,j]-1]\succeq A[i-1,h-1]$ (by Lemma \ref{lem:alg1}(i)) and $\Egal[i-1,k_2[i,j]-1]=\Egal[i-1,h-1]$, we necessarily have that $L[i,k_2[i,j],j]\leq L[i,h,j]$ for any $h\in H_2$. We conclude that the best allocation of type $A[i,h,j]$ with $h\in H_2$ is achieved by $h:=k_2[i,j]$, and this shows the optimality of $A[i,k_2[i,j],j]$.
		\item[(c)] Suppose $k_3[i,j]\leq k$. Then we can also show that $A[i,k_3[i,j],j]\succeq A[i,k,j]$. The set $H_3$ of values $h\geq k_3[i,j]$ with $\Egal[i,h,j]=\Egal[i,j]$ is an interval whose minimum is $k_3[i,j]$, thus both $k$ and $k_3[i,j]$ belong to $H_3$. As $u[k_3[i,j],j]<\Egal[i-1,k_3[i,j]-1]$ (by definition of $k_3[i,j]$), we have that agent $i$ is a loser in $A[i,k_3[i,j],j]$, and it is the only one. Thus, allocation $A[i,k_3[i,j],j]$ must be necessarily optimal. \qedhere
	\end{enumerate}
\end{proof}
By exploiting \eqref{recur_1A} and the characterization of optimal allocations given in Lemma \ref{lem:alg2}, we can derive an efficient dynamic programming algorithm (Algorithm \ref{alg:stronglyMMS}) that computes an SMMS allocation.  
\begin{theorem}\label{thm:stronglyMMS}
	Algorithm \ref{alg:stronglyMMS}  computes an SMMS allocation in $O(mn)$ time.
\end{theorem}
\begin{algorithm}[p]
	\caption{Find an SMMS allocation of a path $P$ with identical and monotonic valuations}
	\label{alg:stronglyMMS}
	\begin{algorithmic}[1]
		\REQUIRE a path $P(v_1,v_m)$, $n$ agents with identical and monotonic valuations $u$ \{$u[h,j]$ denotes  valuation $u(P(v_h,v_j))$, with the convention that $P(v_h,v_j)=\emptyset$ and $u[h,j]=0$ if $h>j$.\}
		\ENSURE an SMMS allocation of $P$
		\FOR{$j = 0,\dots,m$}
		\STATE $\Egal[1,j]\leftarrow u[1,j]$, $L[1,j]\leftarrow 1$
		\ENDFOR
		\FOR{$i = 1,\dots,n$}
		\STATE $\Egal[i,0]\leftarrow 0$, $L[i,0]\leftarrow 0$
		\STATE $k_t[i,0]\leftarrow 1$ for $t=1,2,3$
		\ENDFOR
		\COMMENT{$\Egal[i,j]$ and $L[i,j]$ will be computed in such a way to be respectively the egalitarian welfare and the number of losers of an optimal allocation of subpath $P(v_1,v_j)$ for $i$ agents; each $k_t[i,j]$ will be  defined as in \eqref{k1}, and will be used to compute an index $k[i,j]$ such that $A[i,k[i,j],j]$ is an optimal allocation of subpath $P(v_1,v_j)$ for $i$ agents.}
		\FOR{$i = 2,\dots,n$}
		\FOR{$j = 1,\dots,m$}
		\STATE $k_1[i,j]\leftarrow \max\{\sup\{h\geq k_1[i,j-1]:\Egal[i-1,h-1]<u[h,j]\},k_1[i,j-1]\}$
		\STATE $k_2[i,j]\leftarrow \max\{h\geq k_2[i,j-1]:\Egal[i-1,h-1]\leq u[h,j]\}$
		\COMMENT{By the monotonicity of the utility function, to compute $k_1[i,j]$, with $t=1$  (resp. $t=2$), it is sufficient to analyze in increasing order all the integers $h>k_t[i,j-1]$ until a value $h$ with $\Egal[i-1,h-1]\geq u[h,j]$ (resp. $\Egal[i-1,h-1]>u[h,j]$) is reached; then, $k_t[i,j]$ can be set equal to $h-1$. We conclude that $k_t[i,j]$ can be found in $k_t[i,j]-k_t[i,j-1]$ iterations (for any fixed $i\in [n]$, $j\in [m]$, and $t = 1,2$).}
		\STATE $k_3[i,j]\leftarrow k_2[i,j]+1$
		\STATE $M_t\leftarrow \min\{\Egal[i-1,k_t[i,j]-1],u[k_t[i,j],j]\}$ for $t = 1,2,3$
		\STATE $L_t\leftarrow \begin{cases}L[i-1,h-1]&\text{ if }\Egal[i-1,h-1]<u[h,j],\\
			L[i-1,h-1]+1&\text{ if }\Egal[i-1,h-1]=u[h,j],\\
			1&\text{ if }\Egal[i-1,h-1]>u[h,j].
		\end{cases}$ \\
		\COMMENT{We observe that $M_t$ and $L_t$ are respectively the egalitarian welfare and the number of losers of allocation $A[i,k_t[i,j],j]$.}
		\STATE $t^*\leftarrow$ the index $t$ maximizing $M_t$, and then, minimizing $L_t$ 
		\STATE $\Egal[i,j]\leftarrow M_{t^*}$, $L[i,j]\leftarrow L_{t^*}$, $k[i,j]\leftarrow k_{t^*}[i,j]$
		\COMMENT{$t^*$ is computed in such a way that $A[i,k_{t^*}[i,j],j]$ is the optimal allocation among those of type $A[i,k_t[i,j],j]$, and by Lemma~\ref{lem:alg2}, $A[i,k[i,j],j]$ is an optimal allocation for subpath $P(v_1,v_j)$ and $i$ agents.}
		\ENDFOR
		\ENDFOR
		\STATE $I^n\leftarrow P(v_{k[n,m]},v_{m})$
		\STATE $j\leftarrow k[n,m]-1$
		\FOR{$i=n-1,\ldots,2,1$}
		\STATE $I^i\leftarrow P(v_{k[i,j]},v_{j})$
		\STATE $j \leftarrow k[i,j]-1$
		\ENDFOR
		\RETURN $A=(I^1,\ldots, I^n)$
	\end{algorithmic}
\end{algorithm}
\begin{proof}
	We first show that the output of Algorithm~\ref{alg:stronglyMMS} is an SMMS allocation. In lines 1--5 of the algorithm  we initialize, for $i=1$ or $j=0$, the maximum egalitarian welfare $\Egal[i,j]$ and the number of losers $L[i,j]$ of the optimal allocation $A[i,j]$ (of path $P(v_1,v_j)$ for $i$ agents). By using the characterization provided in Lemma \ref{lem:alg2}, in lines 6--14 we iteratively compute an index $k[i,j]$ such that $A[i,k[i,j],j]$ is an SMMS allocation of subpath $P(v_1,v_j)$ for $i$ agents, and we compute the corresponding maximum egalitarian welfare $\Egal[i,j]$ and number of losers $L[i,j]$. Finally, in lines 15--20 we recursively reconstruct the optimal allocation $A[n,m]$ that is returned as output. 

	Now, we show that the time complexity of Algorithm \ref{alg:stronglyMMS} is $O(m n)$. Observe that the body of the nested for-loops in lines 6--14 can be performed in time $T(i,j)=c\sum_{t=1}^2(k_t[i,j]-k_t[i,j-1])$, where $c$ is a constant that does not depend on $i$ and $j$. Indeed, this running time depends on the computation in lines 8--9 of each index $k_t[i,j]$ for $t = 1,2$, and to compute it we can simply analyze all the indices from $k_t[i,j-1]+1$ to $k_t[i,j]+1$ only. We conclude that the time complexity $T$ of the nested for-loops in lines 6--14 satisfies 
	\begin{align*}
		T&=\sum_{i=1}^n\sum_{j=1}^m T(i,j)=\sum_{i=1}^n\sum_{j=1}^m c\sum_{t=1}^2(k_t[i,j]-k_t[i,j-1])=\sum_{t=1}^2c\sum_{i=1}^n\sum_{j=1}^m (k_t[i,j]-k_t[i,j-1])\\
		&=\sum_{t=1}^2c\sum_{i=1}^n(k_t[i,m]-k_t[i,0])\leq 2c\sum_{i=1}^n m=2cnm\in O(mn).
	\end{align*}
	Since the time complexity of the other parts of the algorithm is clearly $O(mn)$, it follows that Algorithm~\ref{alg:stronglyMMS} terminates in $O(mn)$ time.
\end{proof}

\end{document}